\documentclass[12pt,a4paper]{article}
\usepackage[top=2cm, bottom=2cm, left=2cm, right=2cm]{geometry}
\usepackage{mathptmx} % Times Roman للنصوص والمعادلات
\usepackage{microtype} % تحسين الحروف والمسافات
\usepackage{adjustbox}
\usepackage{amssymb,amsthm,amsmath,latexsym}
\usepackage{authblk}
\usepackage[symbol]{footmisc}
\usepackage{graphicx}
\usepackage[table]{xcolor} % لتلوين الجداول
\usepackage{colortbl} % لتلوين الأعمدة/الصفوف/الخلايا
\usepackage{booktabs} % لتحسين مظهر الجداول
\usepackage{epstopdf} % لتحسين جودة EPS
\usepackage[T1]{fontenc}
\usepackage[utf8]{inputenc}
\usepackage{multirow}
\usepackage{pdfpages}
\usepackage{subcaption}
\usepackage{float}
\usepackage{orcidlink}
\newtheorem {theorem}{Theorem}

\newtheorem {remark}{Remark}
\numberwithin{equation}{section}

\begin{document}
\title{Shiha Distribution: Statistical Properties and Applications to Reliability Engineering and Environmental Data}
\newcommand{\orcidauthorA}{\orcidlink{0000-0001-9539-9488}}% Add \orcidB{} behind the author's name
\author{
\begin{flushleft}
F. A. Shiha$^{}$\orcidauthorA{}\\
{\small
$^{}$Department of Mathematics, Faculty of Science, Mansoura University, 35516 Mansoura, Egypt\\

$^{}$Correspondence: Email: fshiha@mans.edu.eg
}
\end{flushleft}
}
\date{}
\maketitle

\begin{abstract}
This paper introduces a new two-parameter distribution, referred to as the Shiha distribution, which provides a flexible model for skewed lifetime data with either heavy or light tails. The proposed distribution is applicable to various fields, including reliability engineering, environmental studies, and related areas. We derive its main statistical properties, including the moment generating function, moments, hazard rate function, quantile function, and entropy. The stress--strength reliability parameter is also derived in closed form. A simulation study is conducted to evaluate its performance. Applications to several real data sets demonstrate that the Shiha distribution consistently provides a superior fit compared with established competing models, confirming its practical effectiveness for lifetime data analysis.
\end{abstract}



%%% ----------------------------------------------------------------------
\section {Introduction}
A wide range of probability models exists, although some datasets require distributions with greater flexibility and improved fit. Since traditional models do not always capture such patterns, researchers continue to develop and generalize probability distributions to enhance prediction and estimation. This ongoing effort has produced many models capable of adapting to modern reliability data through various construction techniques.
One common approach to increasing flexibility is to construct mixture models by combining two or more distributions with suitable mixing proportions. Such mixtures can represent complex shapes that single distributions often fail to capture. Several three-component mixtures have appeared in the literature. For example, mixing Exp$(\omega)$, Exp$(2\omega)$, and Gamma$(2, 2\omega)$ distributions led to the one-parameter Sarhan–Tadj–Hamilton distribution, studied by Sarhan et al. \cite{sarhan} and later extended through its transmuted version by Shiha et al. \cite{shiha}. Using mixtures of  Exp$(\omega)$, Gamma$(2, \omega)$ , and Gamma$(3, \omega)$  with different mixing proportions, Shanker \cite{shanker2, shanker3} and Welday and Shanker \cite{welday} introduced the Aradhana, Sujatha, and generalized Aradhana distributions.

A classical two-component mixture of Exp$(\omega)$ and Gamma$(2,\omega)$ with mixing proportions $\frac{\omega}{1+\omega}$ and $\frac{1}{1+\omega}$  is the Lindley distribution proposed by Lindley \cite{lind}, which later inspired several extensions, including the power Lindley distribution \cite{ghit}, a generalized Lindley distribution \cite{nadar}, and a three-parameter generalized Lindley distribution \cite{ekho}, among others. Another two-component mixture, based on Exp$(\omega)$ and Gamma$(3,\omega)$ has also been studied, leading to models such as the xgamma distribution \cite{sen}, its alpha-power transformation \cite{shukla}, the Akash distribution \cite{shanker1}, and the Chris–Jerry distribution \cite{chri}.
Although many distributions are widely used in lifetime analysis, there are situations where they do not provide an adequate fit in practice. To address these limitations, we introduce a new two-parameter distribution defined as a mixture of Exp$(\omega)$, Exp$(2\omega)$, and Gamma$(2, 2\omega)$ with appropriate mixing weights. The proposed model has the ability to generate flexible hazard rate patterns, making it suitable for modeling various types of right-skewed lifetime data.
The rest of this article is organized as follows: Section 2 introduces the Shiha distribution and investigates its reliability properties, including the hazard rate function and the stress--strength reliability. Some of its statistical characteristics are studied in Section 3, including the quantile function, moments, skewness, kurtosis, and entropy. Section 4 presents the parameters estimation using the maximum likelihood method. Section 5 applies the proposed distribution to four real lifetime datasets from reliability engineering and environmental studies, and compares its fit with the alpha-power transformed xgamma, power Lindley, three-parameter generalized Lindley, Chris–Jerry, and Akash distributions. Finally,
Section 6 concludes the paper.

\section{The Shiha distribution}
The probability density function (pdf) of the newly proposed two-parameter lifetime distribution is defined as follows:
\begin{equation}\label{pdf}
  f(y;\omega, \eta) = \frac{\omega}{\omega+3\,\eta}\left[\omega + (2\,\eta + 8\,\omega\,\eta y)e^{-\omega y}\right] e^{-\omega y},
\end{equation}

where $ y\geq 0, \, \omega > 0, \, \eta > 0$, we shall refer to this distribution as Shiha distribution ($Y\sim Sh(\omega, \eta))$. Note that the proposed distribution reduces to the $\text{EXP}(\omega)$ distribution when $\eta=0$. This distribution can be represented as a mixture of Exp($\omega$), Exp($2\omega$), and $Gamma(2,2\omega)$ with their mixing proportions
$p_1 = \frac{\omega}{\omega+3 \,\eta}, \quad p_2 = \frac{\eta}{\omega+3\,\eta}, \quad p_3 = \frac{2\eta}{\omega+3\,\eta}$, respectively. This means that
\begin{equation}\label{mix}
  f(y;\omega,\eta) = p_1 \text{Exp}(\omega) + p_2 \text{Exp}(2\omega) + p_3 \text{Gamma}(2, 2\omega).
\end{equation}
The corresponding cumulative distribution function (cdf) of \eqref{pdf} is
\begin{equation}\label{cdf}
  F(y;\omega,\eta) = 1 - \frac{1}{\omega+3\,\eta}\left[\omega + (3\,\eta + 4\,\omega \,\eta y)e^{-\omega y}\right] e^{-\omega y},
\end{equation}
and its survival function is
\begin{equation}\label{sur}
  S(y;\omega,\eta) = \frac{1}{\omega+3\,\eta}\left[\omega + (3\,\eta + 4\,\omega \,\eta y)e^{-\omega y}\right] e^{-\omega y}, \quad y \geq 0, \, \omega > 0, \, \eta > 0.
\end{equation}
While the hazard rate function takes the form:
\begin{equation}\label{hazard}
h(y;\omega,\eta) = \frac{f(y;\omega,\eta)}{S(y;\omega,\eta)}=\frac{\omega \left[\omega + (2\eta + 8\omega\eta y)e^{-\omega y}\right]}{\omega + (3\eta + 4\omega\eta y)e^{-\omega y}}, \quad y \geq 0, \, \omega > 0, \, \eta > 0.
\end{equation}

\begin{theorem}
  The Hazard rate function $h(y;\omega,\eta)$ attains a unique maximum at
\[
y=\frac{1}{\omega}\left[\,W\!\Big(\frac{4\eta}{\omega}e^{-5/4}\Big)+\tfrac{5}{4}\right],
\]
with maximum value
\[
h_{\max} = \omega\;\frac{2W\!\left(\tfrac{4\eta}{\omega}e^{-5/4}\right)+1}
{W\!\left(\tfrac{4\eta}{\omega}e^{-5/4}\right)+1},
\]
where $W(\cdot)>0$ denotes the Lambert $W$ function.
\end{theorem}
\begin{proof}
Let $t=\omega y$, then
\[
h(y;\omega,\eta)=:h(t)
=\omega\;\frac{8\eta t+2\eta+\omega e^{t}}{4\eta t+3\eta+\omega e^{t}}, \qquad t\ge0
\]
The derivatives is
\[
\begin{aligned}
h'(t) &=\frac{ \omega\big[(8\eta+\omega e^{t})(4\eta t+3\eta+\omega e^{t})
     -(8\eta t+2\eta+\omega e^{t})(4\eta+\omega e^{t})\big]}{(4\eta t+3\eta+\omega e^{t})^2} \\[6pt]
     &=\frac{\omega\eta\big(16\eta+\omega e^{t}(5-4t)\big)}{(4\eta t+3\eta+\omega e^{t})^2}.
    \end{aligned}
\]
Since $\omega, \eta >0$, the derivative is equal zero at
\begin{equation}\label{valuet}
\big(t^*-\tfrac{5}{4}\big)e^{t^*}=\frac{4\eta}{\omega}.
\end{equation}
Let \(u=t^*-\tfrac{5}{4}\). Then
\[
u e^{u} = \frac{4\eta}{\omega}e^{-5/4}.
\]
By using the definition of Lambert $W$ function (see Jodrá \cite{jodra} for more details), we get
\[
u = W\!\Big(\frac{4\eta}{\omega}e^{-5/4}\Big),
\qquad
t^{*}=W\!\Big(\frac{4\eta}{\omega}e^{-5/4}\Big)+\frac{5}{4}
\]
Notice that $h'(t)>0 $ for $ 0\le t<t^{*}$ and
$h'(t)<0$  for  $t>t^{*}$, therefore \(h(t)\) attains a unique maximum at \(t^{*}\),
\[
h_{\max}=h(t^{*})=\omega\frac{8\eta t^{*}+2\eta+\omega e^{t^{*}}}{4\eta t^{*}+3\eta+\omega e^{t^{*}}},
\qquad t^*=\omega y^*.
\]
From (\ref{valuet}), we have $ \omega e^{t^*}=\frac{4\eta}{\big(t^*-\tfrac{5}{4}\big)}$, then we obtain
\[
h_{\max}=h(t^*)=\omega\;\frac{2(16(t^*)^2-16t^*+3)}{16(t^*)^2-8t^*+1},
\qquad t^*=\omega y^*.
\]
Substituting $t^{*}=W\!\Big(\frac{4\eta}{\omega}e^{-5/4}\Big)+\frac{5}{4}$, $t^* = \omega y^*$ yields the result.
\end{proof}
\begin{remark} The hazard rate function of Shiha distribution satisfies
\[
h(0;\omega,\eta)=\frac{\omega(\omega+2\, \eta)}{\omega+3\,\eta}< \omega, \qquad \lim_{y \to \infty} h(y;\omega,\eta)=\omega.
\]
Since $W(\cdot)>0$ for positive argument, we have $h_{\max}>\omega$. Moreover,
\[
\lim_{\eta\to0^+} h_{\max}=\omega,\qquad
\lim_{\eta\to\infty} h_{\max}=2\omega,
\]
so $\omega < h_{\max} < 2\omega$ for all $\eta>0 $.
\end{remark}

\begin{figure}[H]
  \centering
  \includegraphics[width=14cm, height=4.5cm]{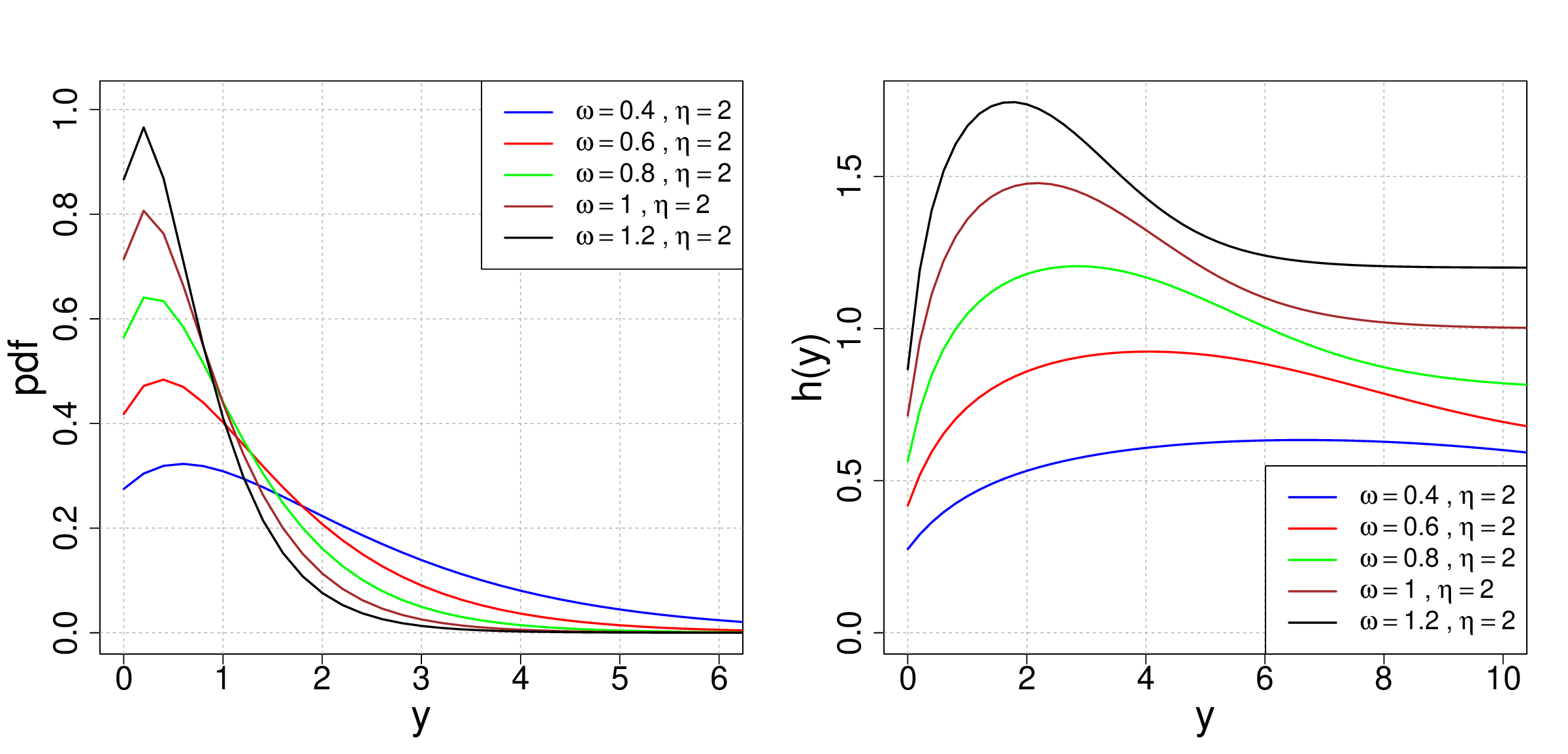}
  \includegraphics[width=14cm, height=4.5cm]{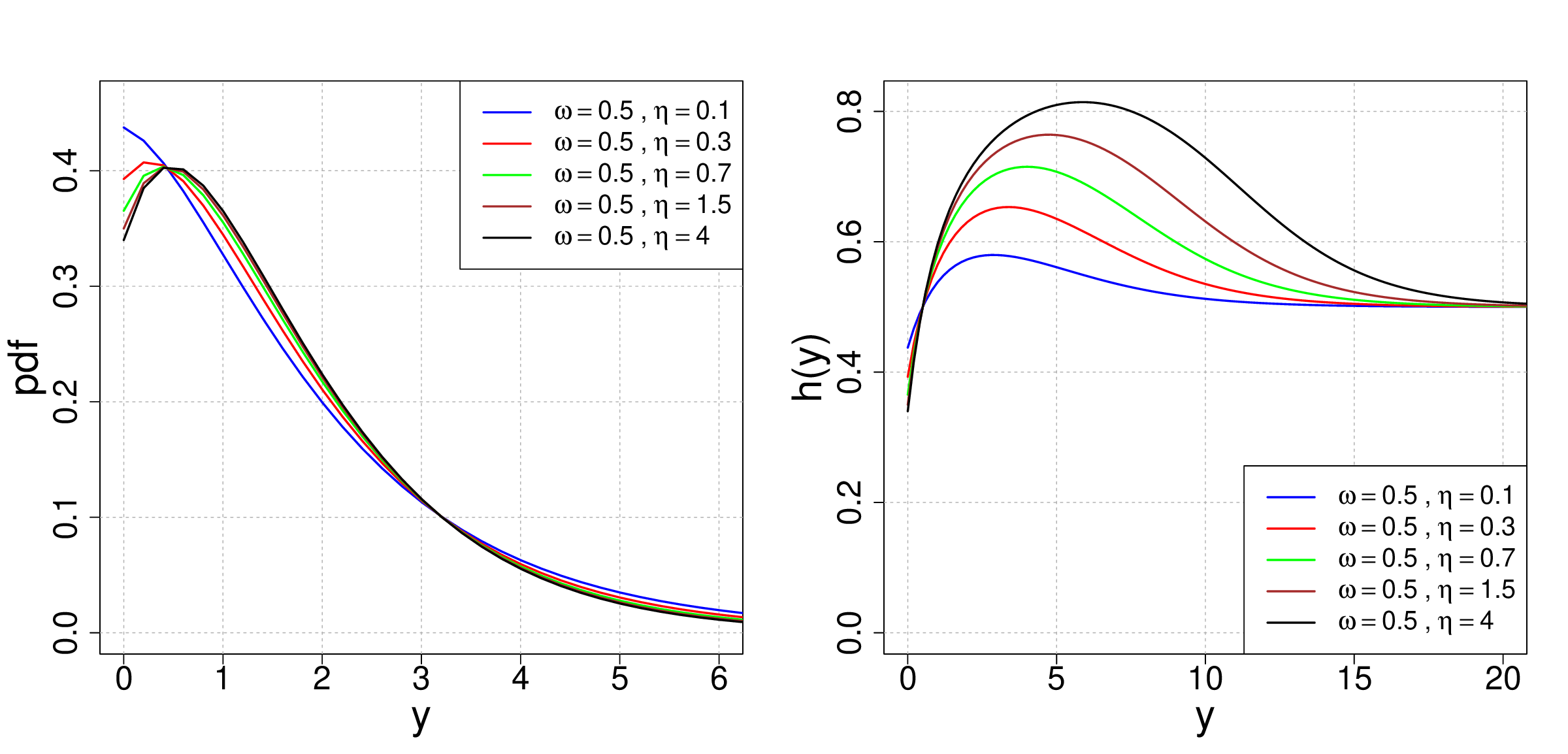}
  \caption{The pdf and hazard functions of the Shiha distribution at different parameter values.}
 \label{fig:pdfhaz}
\end{figure}

Figure \ref{fig:pdfhaz} displays the graphs of the pdf and hazard rate of Shiha distribution for different values of $\omega$ and $\eta$. The figure shows that as $\omega$ increases, the distribution shifts leftward, with the peak becoming sharper and moving toward smaller values of $y$, reflecting a higher concentration of probability near zero. Conversely, increasing $\eta$ flattens the pdf, producing a more gradual tail. The hazard rate exhibits various shapes including increasing, decreasing, unimodal, and upside-down bathtub forms indicating that the distribution can effectively model a wide range of right-skewed lifetime data with either heavy or light tails, making it applicable in engineering and environmental contexts.

\subsection{Stress--Strength Reliability}
The stress--strength reliability parameter is defined as
$R = P(Y_1>Y_2)$, where $Y_1$ (strength) and $Y_2$ (stress) are independent random variables.
The parameter $R$ is widely used in reliability analysis as a measure of system performance.
The system is said to fail whenever the applied stress exceeds its strength, and hence the
failure probability is given by $P(Y_2>Y_1)$.

Let $Y_1$ and $Y_2$ be two independent Shiha random variables with parameters
$(\omega_1, \eta_1)$ and $(\omega_2, \eta_2)$, respectively.
Then the stress--strength reliability parameter $R$ is given by
\[
R =\int_0^{\infty} f(y;\omega_1,\eta_1)\,F(y;\omega_2,\eta_2)\,dy.
\]
Substituting from (\ref{pdf}) and (\ref{cdf}), we obtain
\begin{align*}
R &=1-\frac{\omega_1}{(\omega_1+3\eta_1)(\omega_2+3\eta_2)}
\int_0^\infty
\Big[\omega_1+(2\eta_1+8\omega_1\eta_1 y)e^{-\omega_1 y}\Big]
e^{-(\omega_1+\omega_2)y}  \\
&\qquad\qquad\qquad\qquad\times
\Big[\omega_2+(3\eta_2+4\omega_2\eta_2 y)e^{-\omega_2 y}\Big]
\,dy.
\end{align*}

Evaluating the above integral using
\[
\int_0^\infty y^n e^{-a y}\,dy=\frac{n!}{a^{n+1}}, \quad a>0,
\]
the stress--strength reliability reduces to

\begin{equation}\label{Rfinal}
\begin{aligned}
R &= 1- \frac{\omega_1}{(\omega_1+3\eta_1)(\omega_2+3\eta_2)} \Bigg[\omega_2 \left(\frac{\omega_1}{\omega_1+\omega_2}
+ \frac{2\eta_1}{2\omega_1+\omega_2} + \frac{8\omega_1\eta_1}{(2\omega_1+\omega_2)^2} \right) \\[6pt]
&\qquad\qquad + 3\eta_2 \left( \frac{\omega_1}{\omega_1+2\omega_2} + \frac{\eta_1}{\omega_1+\omega_2} + \frac{2\omega_1\eta_1}{(\omega_1+\omega_2)^2} \right) \\[6pt]
&\qquad\qquad + 4\omega_2\eta_2 \left( \frac{\omega_1}{(\omega_1+2\omega_2)^2} + \frac{\eta_1}{2(\omega_1+\omega_2)^2}+\frac{2\omega_1\eta_1}{(\omega_1+\omega_2)^3}\right)\Bigg].
\end{aligned}
\end{equation}
We now investigate some special cases of \eqref{Rfinal}.
\begin{enumerate}
  \item Setting \(\eta_1 = \eta_2 = 0\), we obtain
  \begin{equation}\label{ssexp}
  R=1- \frac{\omega_1}{\omega_1 \, \omega_2}\left( \frac{\omega_2 \,\omega_1}{\omega_1+\omega_2}\right)=\frac{\omega_2}{{\omega_1+\omega_2}},
\end{equation}
which coincides with the stress--strength reliability of two independent random variables $Y_1 \sim \text{Exp}(\omega_1)$ and $Y_2 \sim \text{Exp}(\omega_2)$.
  \item Setting \(\eta_1 = \eta_2 = \eta\), \(\omega_1 = \omega_2 = \omega\), that is $Y_1\sim Sh(\omega, \eta)$ and $Y_2 \sim Sh(\omega, \eta)$, we get
\begin{equation*}
\begin{aligned}
R &= 1 - \frac{\omega}{(\omega+3\eta)^2} \left( \frac{\omega}{2} + \frac{2\eta}{3} + \frac{8\eta}{9} + \eta + \frac{3\eta^2}{\omega} + \frac{4\eta}{9} + \frac{3\eta^2}{2\omega} \right)\\
&=1-\frac{\omega}{(\omega+3\eta)^2}\left( \frac{\omega}{2} + 3\eta +\frac{9\eta^2}{2\omega}\right)=\frac{1}{2}.
\end{aligned}
\end{equation*}

\item Setting \(\omega_1 = \omega_2 = \omega\), \(\eta_2 = 0\) i.e., $Y_1 \sim Sh(\omega, \eta)$, $Y_2 \sim \text{Exp}(\omega)$, then
\[
R = 1-\frac{1}{\omega+3\eta} \left[ \frac{\omega}{2} +\frac{2\eta}{3} + \frac{8\eta}{9} \right]=\frac{9 \omega+26 \eta}{18(\omega+3\eta)}.
\]

\item If $\omega_1 \to \infty$, the random variable $Y_1$ becomes highly concentrated near zero, and

\[
R \to 1 - \frac{1}{\omega_2+3\eta_2} (\omega_2 \cdot 1 + 3\eta_2 \cdot 1 + 0)
= 1 - \frac{\omega_2+3\eta_2}{\omega_2+3\eta_2} = 0.
\]
\end{enumerate}
\section{Statistical Properties}
In this section, several structural and statistical properties of the proposed Shiha model are investigated, including the quantiles, moment generating function, the first four moments, variance, skewness, kurtosis, and entropy.
\subsection{The Quantile Function}
For a probability $0<p<1$ the $p$-th quantile $y_p$ is defined by
\begin{equation}\label{eq:quantile}
F(y_p;\omega,\eta)=p.
\end{equation}
From \eqref{cdf} and \eqref{eq:quantile} we obtain
\[
1-p=\frac{1}{\omega+3\eta}\left[\omega + \big(3\eta + 4\omega\eta y_p\big)e^{-\omega y_p}\right] e^{-\omega y_p},
\]
which can be rearranged to
\begin{equation}\label{eq:final}
\omega e^{-\omega y_p} + 3\eta e^{-2\omega y_p} + 4\omega\eta y_p e^{-2\omega y_p} - (1-p)(\omega+3\eta) = 0.
\end{equation}
The substitution \(x=e^{-\omega y_p}\) (so \(y_p=-\frac{1}{\omega}\ln x\)) transforms \eqref{eq:final} into
\begin{equation}\label{eq:fquantile}
\omega x + 3\eta x^2 - 4\eta x^2\ln x = (1-p)(\omega+3\eta),\qquad 0<x<1.
\end{equation}

Equation \eqref{eq:fquantile} is non-algebraic  due to the mixture of polynomial and logarithmic terms, so no closed form solution exists for $x$ (and hence for $y_p$).  The quantiles were obtained numerically by solving $ F(y;\omega,\eta)-p=0$ using R program software. Since the F($y;\omega,\eta)$ is continuous and strictly increasing, a unique quantile $y_p$ exists for each $p \in (0,1)$.
Table~\ref{tab:quantiles} presents selected quantiles for different combinations of $\omega,\eta$, while Figure~\ref{fig:qua} illustrates their variation across probability levels.

\begin{table}[h!]
\centering
\caption{Quantile values for different parameter values.}
\label{tab:quantiles}
\begin{tabular}{|c|c|c|c|c|}
\hline
Probability & $\boldsymbol{\omega=0.4, \eta=0.2}$ & $\boldsymbol{\omega=0.4, \eta=3}$ & $\boldsymbol{\omega=1.2, \eta=0.2}$ & $\boldsymbol{\omega=1.2, \eta=3}$ \\
\hline
0.01 & 0.0311 & 0.0363 & 0.0094 & 0.0117 \\
0.05 & 0.1544 & 0.1741 & 0.0473 & 0.0566 \\
0.10 & 0.3069 & 0.3363 & 0.0955 & 0.1100 \\
0.25 & 0.7736 & 0.8036 & 0.2500 & 0.2658 \\
0.40&   1.2886& 1.2943&  0.4279 &0.4310\\
0.42 &  1.3635& 1.3643 & 0.4543& 0.4547\\
0.43  & 1.4016& 1.3998 & 0.4677 &0.4667\\
0.45 &  1.4795 &1.4722 & 0.4952& 0.4912\\
0.50 & 1.6841 & 1.6611 & 0.5680 & 0.5552 \\
0.60  & 2.1495 &2.08456 & 0.7355& 0.6991\\
0.75 & 3.0927 & 2.9203 & 1.0815 & 0.9846 \\
0.90 & 4.8874 & 4.4343 & 1.7547 & 1.5079 \\
0.95 & 6.2642 & 5.5294 & 2.2759 & 1.8925 \\
0.99 & 9.6642 & 8.0174 & 3.5386 & 2.7954 \\
\hline
\end{tabular}
\end{table}

\begin{figure}[H]
  \centering
 \includegraphics[width=14cm, height=5cm]{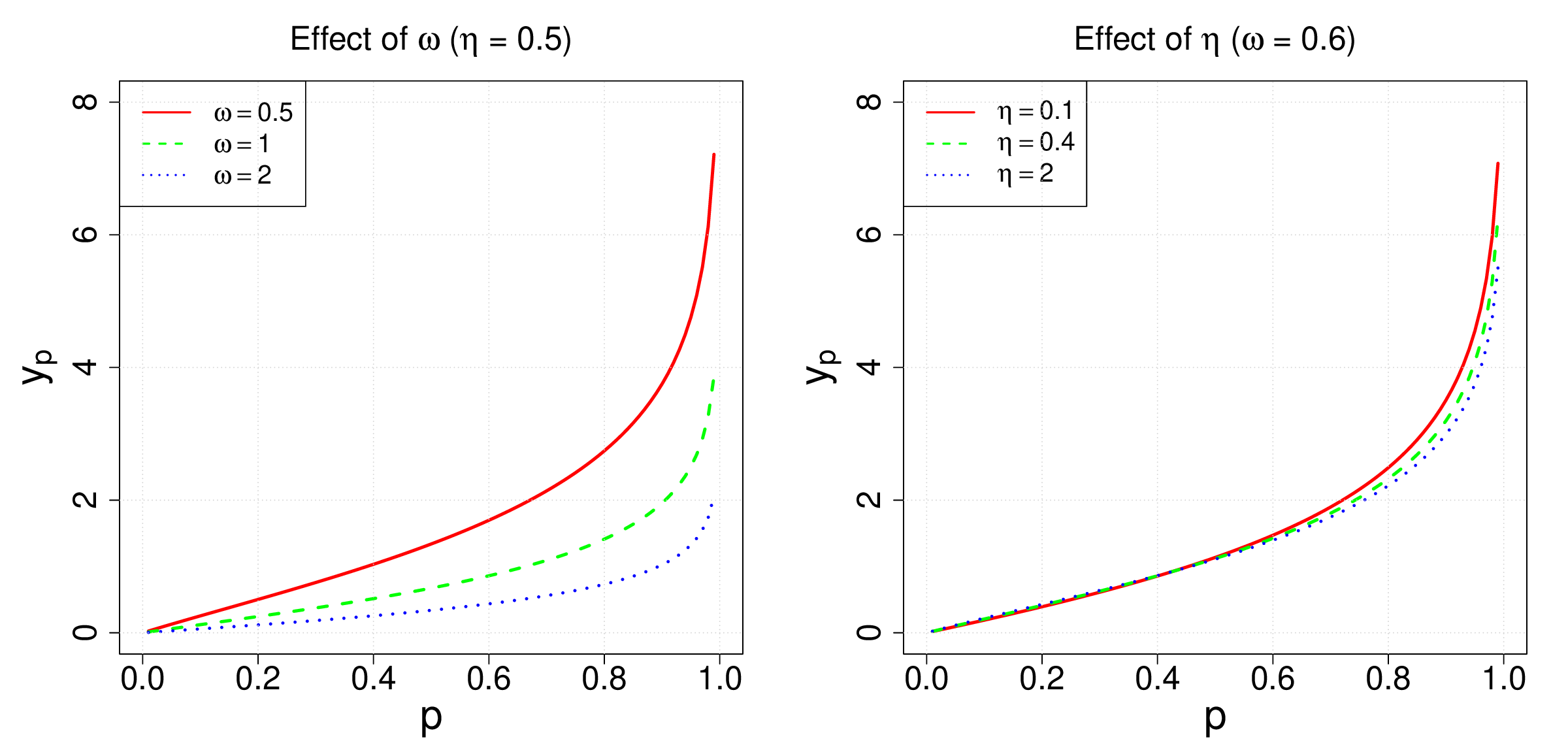}
  \caption{The quantiles of Shiha distribution.}
 \label{fig:qua}
\end{figure}
From Table~\ref{tab:quantiles} and Figure~\ref{fig:qua}, it can be seen that for a fixed $\omega$, increasing $\eta$ slightly raises the quantiles for lower probabilities ($p\in(0,0.43]$), while they produce smaller quantiles for higher probabilities ($p\in(0.43,1)$).
This indicates that $\eta$ has a non-monotonic effect on the quantile function: it increases the lower part and compresses the upper tail of the distribution.
Conversely, higher $\omega$ values reduce all quantiles, confirming its role as a scale parameter controlling the overall magnitude of the distribution.
To further illustrate the parameter effects, Figure~\ref{fig:qur} shows the quartiles trends.
\begin{figure}[H]
  \centering
  \includegraphics[width=14cm, height=4.5cm]{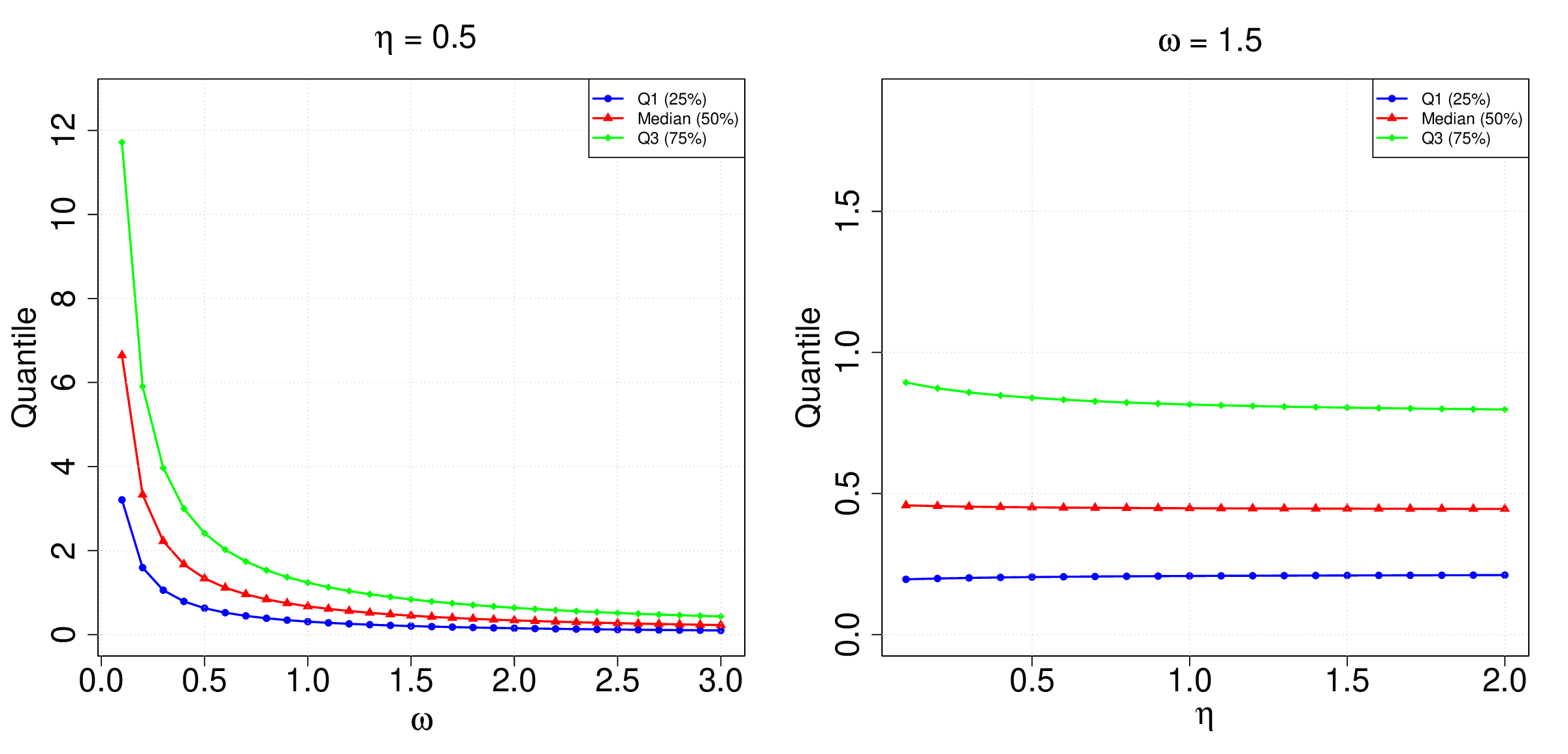}
  \caption{The quartiles of the Shiha distribution.}
 \label{fig:qur}
\end{figure}

As $\omega$ increases, all quartiles decline and the interquartile range $Q_3 - Q_1$ narrowers, indicating reduced dispersion and a leftward shift of the distribution toward smaller values. Meanwhile, variations in $\eta$ have only a slight effect on the quartiles, suggesting that $\eta$ primarily affects the tails rather than the central spread.

\subsection{Moment Generating Function}
The moment generating function (MGF) captures all moments of a distribution, and the derived skewness and kurtosis describe its asymmetry and tail characteristics.

\begin{theorem}
The moment generating function of a Shiha distributed random variable is given by

\begin{equation}\label{eq:mgf}
M_Y(t) = \frac{4\omega^{4} + 12\omega^{3}\eta - 4\omega^{3}t - 14\omega^{2}\eta t + \omega^{2}t^{2} + 2\omega\eta t^{2}}
{(\omega + 3\eta)(\omega - t)(2\omega - t)^{2}}, \qquad t < \omega.
\end{equation}
\end{theorem}
\begin{proof}
By definition,
\[
M_Y(t)=E[e^{tY}]=\int_{0}^{\infty} e^{ty} f(y;\omega,\eta)\,dy.
\]

Substituting $f(y;\omega,\eta)$ from \eqref{pdf} and simplifying the exponentials gives
\[
\begin{aligned}
M_Y(t)
&=\frac{\omega}{\omega+3\eta}\int_{0}^{\infty} e^{ty}\Big[\omega + (2\eta+8\omega\eta y)e^{-\omega y}\Big] e^{-\omega y}\,dy\\[4pt]
&=\frac{\omega}{\omega+3\eta}\int_{0}^{\infty}\Big[\omega e^{(t-\omega)y} + (2\eta+8\omega\eta y)e^{(t-2\omega)y}\Big]\,dy \\[4pt]
&=\frac{\omega}{\omega+3\eta} \left[ \omega\int_{0}^{\infty} e^{(t-\omega)y}\,dy+2\eta\int_{0}^{\infty} e^{(t-2\omega)y}\,dy+8\omega\eta\int_{0}^{\infty} y\,e^{(t-2\omega)y}\,dy \right] \\
&= \frac{\omega}{\omega+3\eta} \left[ \frac{\omega}{\omega-t}+\frac{2\eta}{2\omega-t}+\frac{8\omega\eta}{(2\omega-t)^2}\right] ,\qquad t<\omega,
\end{aligned}
\]
which is equivalent to (\ref{eq:mgf}).
\end{proof}
The $k$th moment about the origin, $\mu_{k}' = E(Y^{k})$, is obtained by taking the $k$th derivative of $M_Y(t)$ with respect to $t$ and then substituting $t = 0$.
\begin{equation}
\mu_{k}' = \frac{\left( 2^{k}\omega + 2\eta k + 3\eta \right) \Gamma(k+1)}
{2^{k} \omega^{k} (\omega + 3\eta)}, \qquad k = 1, 2, 3, \ldots
\end{equation}
By setting $k = 1, 2, 3,$ and $4$, we obtain the first four moments of the Shiha distribution.
\[
\begin{aligned}
\mu'_1 &= \frac{2\omega+5\eta}{2\omega(\omega+3\eta)},  & \mu'_2 &= \frac{4\omega+7\eta}{2\omega^2(\omega+3\eta)},
 & \mu'_3 & = \frac{3(8\omega+9\eta)}{4\omega^3(\omega+3\eta)}, &  \mu'_4 &= \frac{3(16\omega+11\eta)}{2\omega^4(\omega+3\eta)}.
\end{aligned}
\]
The numerical values of the first four moments of Shiha distribution for several values of $\eta$ and $\omega$ are shown in Table~\ref{T:mom}. It can be seen that the moments decrease as the parameters increase.

\begin{table}[H]
\centering
\caption{The first four moments of the Shiha distribution for different values of the model parameters.}\label{T:mom}
\adjustbox{max width=\textwidth, center}{\footnotesize
\begin{tabular}{c|cccc|cccc|cccc}\hline
\multirow{2}{*}{$\omega$}
& \multicolumn{4}{c|}{$\eta=0.2$}
& \multicolumn{4}{c|}{$\eta=0.6$}
& \multicolumn{4}{c}{$\eta=1$} \\ \cline{2-13}
 & $\mu'_1$ & $\mu'_2$ & $\mu'_3$ & $\mu'_4$ & $\mu'_1$ & $\mu'_2$ & $\mu'_3$ & $\mu'_4$ & $\mu'_1$ & $\mu'_2$ & $\mu'_3$ & $\mu'_4$ \\\hline

 0.2 & 4.375 & 34.375 & 398.437 & 6328.125 & 4.25 & 31.25 & 328.125 & 4593.75 & 4.219 & 30.469 & 310.547 & 4160.156 \\
0.4 & 2.25 & 9.375 & 58.594 & 503.906 & 2.159 & 8.239 & 45.81 & 346.236 & 2.132 & 7.904 & 42.05 & 299.862 \\
0.6 & 1.528 & 4.398 & 19.097 & 113.812 & 1.458 & 3.819 & 14.757 & 78.125 & 1.435 & 3.627 & 13.31 & 66.229 \\
0.8 & 1.161 & 2.567 & 8.58 & 39.237 & 1.106 & 2.224 & 6.648 & 27.325 & 1.086 & 2.097 & 5.936 & 22.936 \\
1 & 0.938 & 1.688 & 4.594 & 17.062 & 0.893 & 1.464 & 3.589 & 12.107 & 0.875 & 1.375 & 3.188 & 10.125 \\
1.2 & 0.787 & 1.196 & 2.749 & 8.6 & 0.75 & 1.042 & 2.17 & 6.221 & 0.734 & 0.976 & 1.922 & 5.201 \\
1.4 & 0.679 & 0.893 & 1.777 & 4.803 & 0.647 & 0.781 & 1.418 & 3.539 & 0.633 & 0.731 & 1.255 & 2.964 \\
1.6 & 0.597 & 0.692 & 1.215 & 2.892 & 0.57 & 0.609 & 0.98 & 2.168 & 0.557 & 0.569 & 0.868 & 1.821 \\
1.8 & 0.532 & 0.553 & 0.868 & 1.846 & 0.509 & 0.489 & 0.707 & 1.405 & 0.498 & 0.457 & 0.627 & 1.185 \\
2 & 0.481 & 0.452 & 0.642 & 1.233 & 0.461 & 0.401 & 0.528 & 0.952 & 0.45 & 0.375 & 0.469 & 0.806 \\
2.2 & 0.438 & 0.376 & 0.488 & 0.855 & 0.42 & 0.336 & 0.405 & 0.669 & 0.411 & 0.314 & 0.36 & 0.569 \\
2.4 & 0.403 & 0.318 & 0.38 & 0.612 & 0.387 & 0.285 & 0.318 & 0.484 & 0.378 & 0.267 & 0.283 & 0.414 \\
2.6 & 0.373 & 0.273 & 0.301 & 0.449 & 0.358 & 0.245 & 0.254 & 0.36 & 0.35 & 0.23 & 0.227 & 0.308 \\
2.8 & 0.347 & 0.236 & 0.243 & 0.337 & 0.334 & 0.214 & 0.206 & 0.273 & 0.326 & 0.2 & 0.185 & 0.235 \\
3 & 0.324 & 0.207 & 0.199 & 0.258 & 0.312 & 0.187 & 0.17 & 0.211 & 0.306 & 0.176 & 0.153 & 0.182 \\
3.2 & 0.304 & 0.182 & 0.165 & 0.201 & 0.294 & 0.166 & 0.142 & 0.165 & 0.287 & 0.156 & 0.128 & 0.144 \\
3.4 & 0.287 & 0.162 & 0.138 & 0.159 & 0.277 & 0.148 & 0.12 & 0.132 & 0.271 & 0.139 & 0.108 & 0.115 \\
3.6 & 0.271 & 0.145 & 0.117 & 0.127 & 0.262 & 0.133 & 0.102 & 0.106 & 0.257 & 0.125 & 0.092 & 0.093 \\
3.8 & 0.257 & 0.131 & 0.1 & 0.103 & 0.249 & 0.12 & 0.087 & 0.087 & 0.244 & 0.113 & 0.079 & 0.076 \\
4 & 0.245 & 0.118 & 0.086 & 0.084 & 0.237 & 0.109 & 0.076 & 0.071 & 0.232 & 0.103 & 0.069 & 0.063 \\ \hline

\end{tabular}}
\end{table}
Skewness and kurtosis are important measures that describe the shape and spread of a probability distribution. They provide useful information about asymmetry, peakedness, and possible outliers in the data. Using the first four moments of the Shiha distribution, the skewness $Sk(Y)$ and kurtosis $Ku(Y)$ of the Shiha distribution are derived as:

The variance of Shiha distribution is
\[
\text{Var}(Y)=\mu'_2 -( \mu'_1)^2=\frac{17\eta^{2}+18\eta\omega+4\omega^{2}}{4\omega^{2}(\omega+3\eta)^{2}}.
\]
The skewness and kurtosis are
\[
\text{Sk(Y)}=\frac{\mu'_3 - 3\mu'_2 \, \mu'_1 + 2(\mu'_1)^3}{{(\text{Var}(Y))}^{3/2}}
=\frac{106\eta^{3}+234\eta^{2}\omega+114\eta\omega^{2}+16\omega^{3}}
{\bigl(17\eta^{2}+18\eta\omega+4\omega^{2}\bigr)^{3/2}},
\]
\[
\text{Ku(Y)}=\frac{\mu'_4 - 4\mu'_3 \, \mu'_1 + 6\mu'_2 \, (\mu'_1)^2 - 3{\mu'_1}^4}{{(\mathrm{Var}(Y))}^{2}}
=\frac{3\bigl(611\eta^{4}+2076\eta^{3}\omega+1608\eta^{2}\omega^{2}+472\eta\omega^{3}+48\omega^{4}\bigr)}
{\bigl(17\eta^{2}+18\eta\omega+4\omega^{2}\bigr)^{2}},
\]
\[
\text{Excess kurtosis}=\text{Ku(Y)}-3
=\frac{966\eta^{4}+4392\eta^{3}\omega+3444\eta^{2}\omega^{2}+984\eta\omega^{3}+96\omega^{4}}
{\bigl(17\eta^{2}+18\eta\omega+4\omega^{2}\bigr)^{2}}.
\]
We note that
\[
\lim_{\omega\to\infty}\mathrm{Sk}(Y)=2,\qquad
\lim_{\omega\to\infty}\mathrm{Ku}(Y)=9.
\]
Table~\ref{T:mom2} presents the computed values of $\sigma^2$, $Sk(Y)$, and $Ku(Y)$. In addition, Figure~\ref{fig:sku} displays the variation of $Sk(Y)$ and $Ku(Y)$ with different values of $\omega$ and $\eta$.

\begin{table}[!ht]
\centering
\caption{Variance, skewness, and kurtosis for different values of $\omega$ and $\eta$.}\label{T:mom2}
\adjustbox{max width=\textwidth, center}{
\begin{tabular}{c|ccc|ccc|ccc|ccc}
\hline
\multirow{2}{*}{$\omega$}
& \multicolumn{3}{c|}{$\eta=0.2$}
& \multicolumn{3}{c|}{$\eta=0.6$}
& \multicolumn{3}{c|}{$\eta=1$}
& \multicolumn{3}{c}{$\eta=1.5$} \\ \cline{2-13}
& $\sigma^2$ & Sk($Y$) & Ku($Y$)
& $\sigma^2$ & Sk($Y$) & Ku($Y$)
& $\sigma^2$ & Sk($Y$) & Ku($Y$)
& $\sigma^2$ & Sk($Y$) & Ku($Y$) \\ \hline
0.2 & 15.234 & 1.93 & 9.497 & 13.187 & 1.738 & 8.186 & 12.671 & 1.665 & 7.618 & 12.395 & 1.621 & 7.26 \\
0.4 & 4.312 & 2.02 & 9.917 & 3.577 & 1.859 & 9.054 & 3.357 & 1.768 & 8.412 & 3.234 & 1.704 & 7.923 \\
0.6 & 2.064 & 2.047 & 9.942 & 1.693 & 1.93 & 9.497 & 1.567 & 1.84 & 8.925 & 1.492 & 1.768 & 8.412 \\
0.8 & 1.22 & 2.055 & 9.886 & 1.001 & 1.974 & 9.731 & 0.919 & 1.892 & 9.266 & 0.868 & 1.818 & 8.777 \\
1 & 0.809 & 2.057 & 9.815 & 0.667 & 2.002 & 9.855 & 0.609 & 1.93 & 9.497 & 0.572 & 1.859 & 9.054 \\
1.2 & 0.577 & 2.056 & 9.748 & 0.479 & 2.02 & 9.917 & 0.437 & 1.958 & 9.655 & 0.408 & 1.892 & 9.266 \\
1.4 & 0.432 & 2.054 & 9.687 & 0.362 & 2.033 & 9.944 & 0.33 & 1.98 & 9.763 & 0.307 & 1.918 & 9.43 \\
1.6 & 0.337 & 2.051 & 9.634 & 0.284 & 2.041 & 9.95 & 0.259 & 1.997 & 9.836 & 0.241 & 1.94 & 9.556 \\
1.8 & 0.27 & 2.049 & 9.587 & 0.229 & 2.047 & 9.942 & 0.209 & 2.01 & 9.886 & 0.194 & 1.958 & 9.655 \\
2 & 0.221 & 2.046 & 9.547 & 0.189 & 2.051 & 9.927 & 0.172 & 2.02 & 9.917 & 0.16 & 1.974 & 9.731 \\
2.2 & 0.184 & 2.044 & 9.511 & 0.159 & 2.054 & 9.908 & 0.145 & 2.029 & 9.937 & 0.135 & 1.986 & 9.791 \\
2.4 & 0.156 & 2.042 & 9.479 & 0.136 & 2.055 & 9.886 & 0.124 & 2.035 & 9.947 & 0.115 & 1.997 & 9.836 \\
2.6 & 0.134 & 2.04 & 9.451 & 0.117 & 2.056 & 9.863 & 0.107 & 2.04 & 9.95 & 0.099 & 2.006 & 9.871 \\
2.8 & 0.116 & 2.038 & 9.426 & 0.102 & 2.057 & 9.839 & 0.094 & 2.044 & 9.948 & 0.087 & 2.014 & 9.898 \\
3 & 0.102 & 2.036 & 9.404 & 0.09 & 2.057 & 9.815 & 0.083 & 2.047 & 9.942 & 0.077 & 2.02 & 9.917 \\
3.2 & 0.09 & 2.034 & 9.384 & 0.08 & 2.057 & 9.792 & 0.073 & 2.05 & 9.934 & 0.068 & 2.026 & 9.931 \\
3.4 & 0.08 & 2.033 & 9.365 & 0.071 & 2.056 & 9.769 & 0.066 & 2.052 & 9.924 & 0.061 & 2.031 & 9.941 \\
3.6 & 0.072 & 2.032 & 9.348 & 0.064 & 2.056 & 9.748 & 0.059 & 2.053 & 9.912 & 0.055 & 2.035 & 9.947 \\
3.8 & 0.064 & 2.03 & 9.333 & 0.058 & 2.055 & 9.727 & 0.054 & 2.055 & 9.899 & 0.05 & 2.038 & 9.949 \\
4 & 0.058 & 2.029 & 9.319 & 0.053 & 2.054 & 9.706 & 0.049 & 2.055 & 9.886 & 0.045 & 2.041 & 9.95 \\
4.2 & 0.053 & 2.028 & 9.306 & 0.048 & 2.054 & 9.687 & 0.045 & 2.056 & 9.872 & 0.042 & 2.044 & 9.948 \\
4.4 & 0.049 & 2.027 & 9.294 & 0.044 & 2.053 & 9.668 & 0.041 & 2.057 & 9.858 & 0.038 & 2.046 & 9.944 \\
4.6 & 0.045 & 2.026 & 9.283 & 0.041 & 2.052 & 9.651 & 0.038 & 2.057 & 9.844 & 0.035 & 2.048 & 9.94 \\
4.8 & 0.041 & 2.025 & 9.273 & 0.037 & 2.051 & 9.634 & 0.035 & 2.057 & 9.829 & 0.033 & 2.05 & 9.934 \\
5 & 0.038 & 2.025 & 9.264 & 0.035 & 2.05 & 9.618 & 0.032 & 2.057 & 9.815 & 0.03 & 2.051 & 9.927 \\ \hline

\hline
\end{tabular}}
\end{table}

Table~\ref{T:mom2} and Figure~\ref{fig:sku} illustrate that both skewness and kurtosis are positive, confirming the right-skewed and heavy-tailed nature of the Shiha distribution. As $\omega$ increases, both measures initially rise, reach a peak, and then decline, with a sharper decline for smaller values of $\eta$. Furthermore, as $\eta$ increases, skewness and kurtosis generally decrease at lower $\omega$ values, but tend to increase slightly again for larger $\omega$. This indicates that $\eta$ acts as a stabilizing parameter at small $\omega$, while for large $\omega$ it amplifies the tail thickness and asymmetry to some extent, reflecting the complex interaction between $\eta$ and $\omega$ in shaping the distribution.
\begin{figure}[H]
  \centering
 \includegraphics[width=14cm, height=4.5cm]{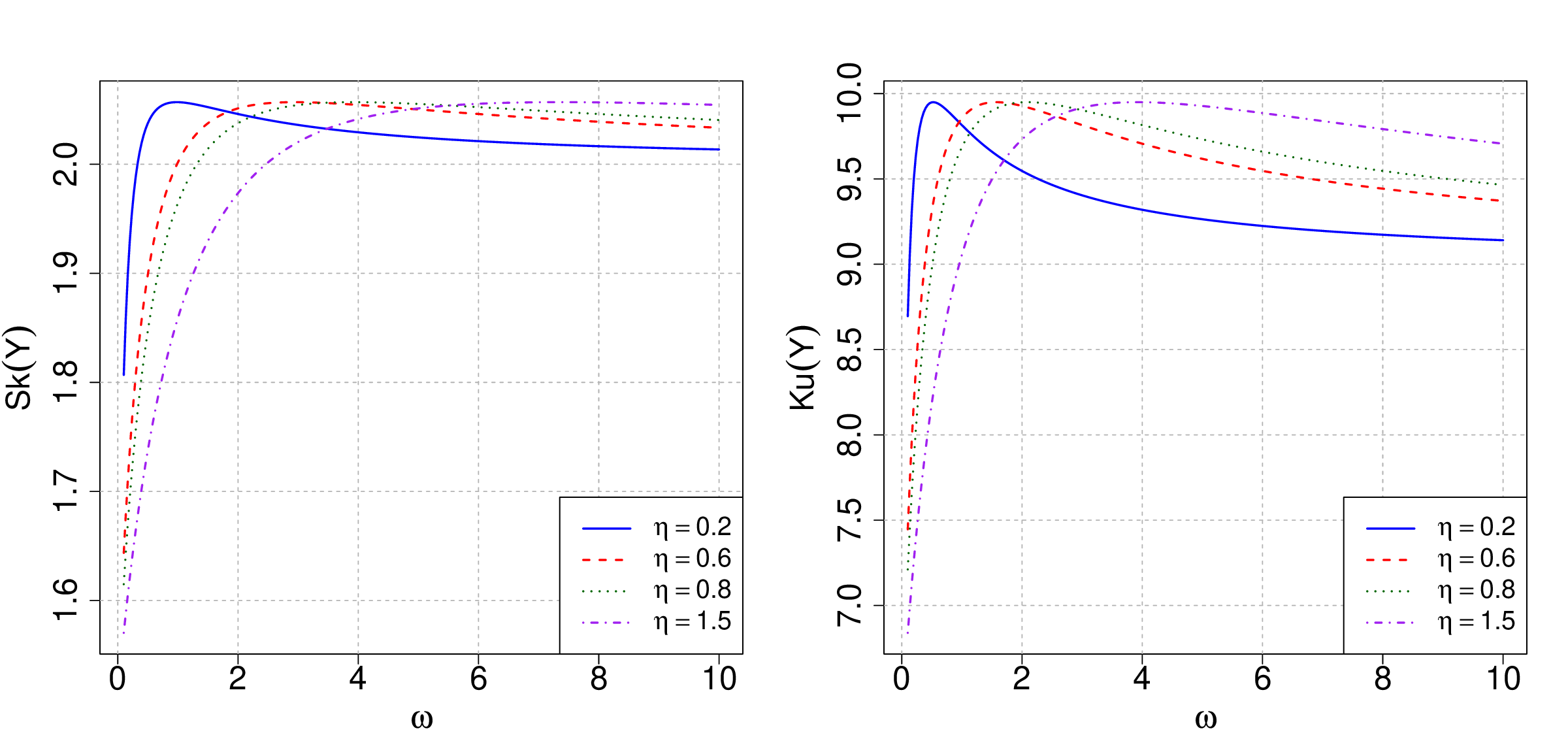}
  \caption{The skewness and kurtosis measures of the Shiha distribution.}
 \label{fig:sku}
\end{figure}

\subsection{Entropy}
 Entropy serves as an important measure of uncertainty and has wide applications in statistics, physics, and machine learning. For a continuous random variable $Y$ with a pdf $f(y) $, the Shannon entropy (differential entropy) is defined as $ H(y)=E[-\ln f(y)]$.

\begin{theorem}
The Shannon entropy of Shiha distribution is
\begin{equation}\label{fentropy}
  H(y)
  = \ln\!\frac{\omega+3\eta}{\omega}
  + \frac{2\omega+5\eta}{2(\omega+3\eta)}
  - E\!\left[\ln\!\big(\omega+(2\eta+8\omega\eta Y)e^{-\omega Y}\big)\right].
\end{equation}
\end{theorem}
\begin{proof}
Let $Y$ be a continuous random variable having a pdf $f(y;\omega,\eta)$. The Shannon entropy of
$Y$ is given by
\begin{equation} \label{entropy}
H(y)=E[-\ln f(y;\omega,\eta)]=- \int_{0}^{\infty} f(y;\omega,\eta)\, \ln f(y;\omega,\eta)\, dy.
\end{equation}
Using the pdf of Shiha distribution, given in (\ref{pdf}), we get

\begin{equation}\label{lnf}
  -\ln f(y;\omega,\eta)
  = -\ln\omega + \ln(\omega+3\eta)
  - \ln\!\big[\omega+(2\eta+8\omega\eta y)e^{-\omega y}\big]
  + \omega y.
\end{equation}
Taking the expectation,
\begin{equation}
  H(y)= -\ln\omega + \ln(\omega+3\eta)
   - E\!\left[\ln\!\big(\omega+(2\eta+8\omega\eta Y)e^{-\omega Y}\big)\right]
   + \omega\,E[Y].
\end{equation}

Given that $E[Y]=\mu'_1=\frac{2\omega+5\eta}{2\omega(\omega+3\eta)}$, we obtain (\ref{fentropy}).
\end{proof}
The expected values in  (\ref{fentropy}) cannot be obtained analytically, $H(y)$ is computed using numerical method.
R programming can be used to compute the entropy for different values of the model parameters, as shown in Figure \ref{fig:entropy}.

\begin{figure}[htbp]
  \centering
 \includegraphics[width=14cm, height=5cm]{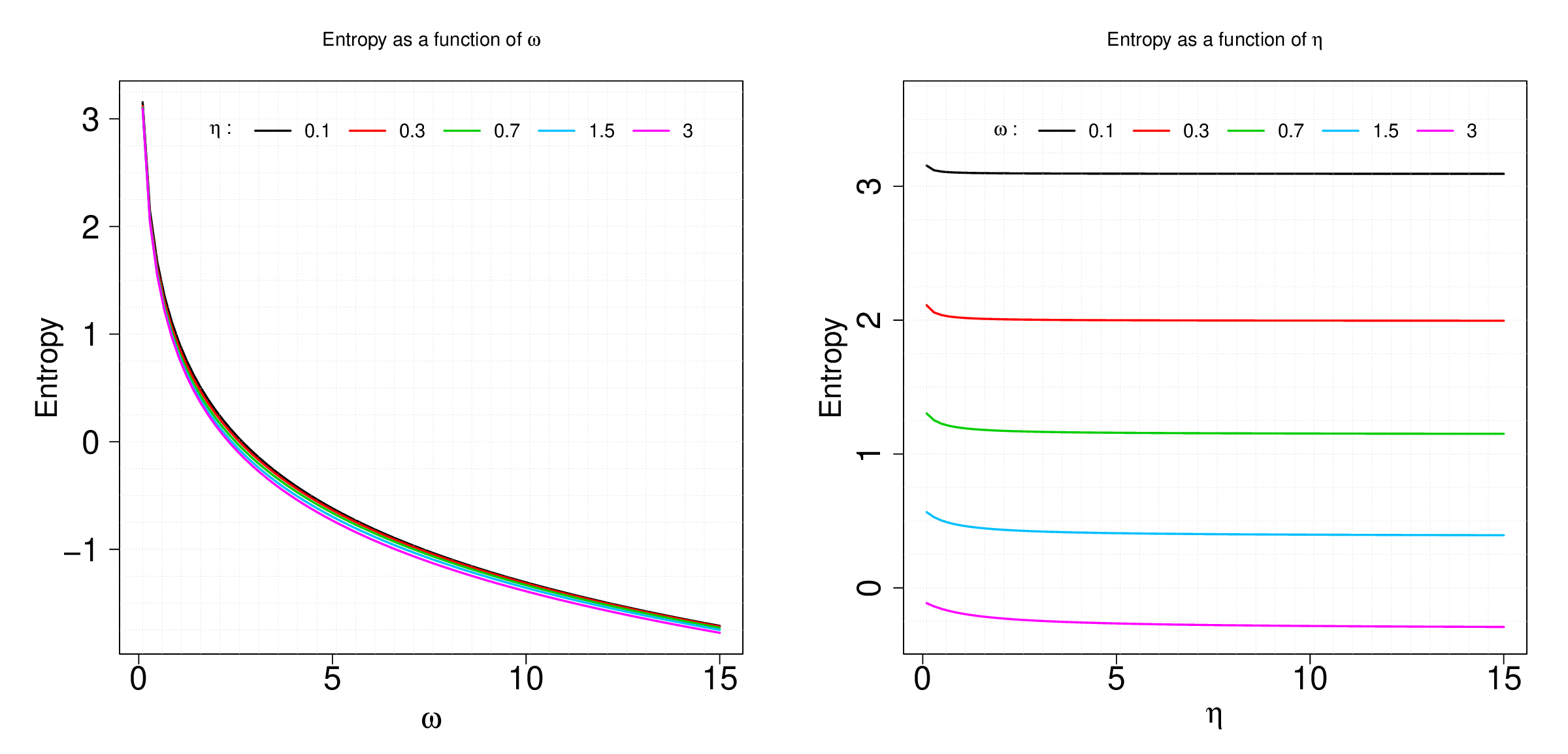}
  \caption{Entropy of the Shiha distribution.}
 \label{fig:entropy}
\end{figure}
For Shiha distribution, the entropy decreases as $\omega$ increases, indicating that larger $\omega$ values make the distribution more concentrated (more peaked and less dispersed) and less uncertain (more predictable). The parameter $\eta$ appears to have a minimal effect on uncertainty of the distribution for the examined parameter ranges, whereas $\omega$  primarily governs the entropy behavior. Note that the differential entropy for continuous distributions can be negative.

\section{Random Sample Generation and Simulation}
The following two equivalent algorithms can be used to generate a random variable from the Shiha distribution. First, set the parameter values $\omega, \eta$.

\textbf{Algorithm 1} (Inverse cdf method).
\begin{enumerate}
\item Generate $u\sim U(0, 1)$.
\item Solve numerically the following equation in $x\in (0, 1) $
\begin{equation}\label{eq:fqua}
\omega x + 3\eta x^2 - 4\eta x^2\ln x -(\omega+3\eta)u=0.
\end{equation}
\item Set $Y=(-\ln x)/\omega $.
\end{enumerate}

\textbf{Algorithm 2} (Mixture form of three distributions). Compute the mixture weights $p_1, p_2$, and $p_3$.
\begin{enumerate}
  \item Generate $u \sim U(0,1)$.
  \item Generate $v_{1} \sim$ Exp$(\omega)$, $v_{2} \sim$ Exp$(2 \omega)$, and  $v_{3}\sim$ Gamma$(2, 2\omega)$.
   \item If $u \leq p_1$, set $Y=v_{1}$, if $p_1 < u \leq p_1+p_2$, set $Y=v_{2}$,
     otherwise, set $Y=v_{3}$.
\end{enumerate}

In the subsequent simulation study, random samples are generated from the Shiha distribution. For each sample size $n = 30, 50, 100, 200, 300,$ and $600$, the simulation is repeated $N = 10000$ times to ensure the stability and reliability of the results.

A Monte Carlo simulation study is then conducted to evaluate the performance of the maximum likelihood estimators for the parameters of the proposed distribution. The estimators $\hat{\omega}$ and $\hat{\eta}$ are assessed based on the criteria: bias (estimation accuracy) and mean squared error (MSE). The bias and MSE of the estimator are computed as
\[
\text{Bias}(\hat{\theta}) = \frac{1}{N}\sum_{i=1}^N \hat{\theta}_i - \theta,
\quad
\text{MSE}(\hat{\theta}) = \frac{1}{N}\sum_{i=1}^N \left( \hat{\theta}_i - \theta \right)^2,
\]
where $\hat{\theta}_i$ denotes the estimate of a parameter (either $\omega$ or $\eta$) obtained in the $i$-th replication, and $\theta$ represents the corresponding true parameter value.
This evaluation provides robust insights into the performance of the estimators across different sample sizes. The simulation results are summarized in Tables~\ref{T:sim4}. All estimators are consistent, as both the bias and MSE of $\hat{\omega}$ and $\hat{\eta}$ decrease with increasing sample size, confirming the consistency and reliability of the MLE method for the Shiha distribution.

\begin{table}[!ht]
\centering
\caption{ The bias and MSE for estimated $\omega$ and $\eta$.}\label{T:sim4}
\adjustbox{max width=\textwidth, center}{\scriptsize
\begin{tabular}{| c|c c  c c|c |c c c c|}
\hline
 \multicolumn{5}{|c|}{$\omega=0.5$, $\eta=0.5$} & \multicolumn{5}{|c|}{$\omega=1$, $\eta=0.5$} \\ \hline
  $n$ & $\text{Bias}(\hat{\omega})$ & $\text{Bias}(\hat{\eta})$ & $\text{MSE}(\hat{\omega})$ & $\text{MSE}(\hat{\eta})$ &$n$ & $\text{Bias}(\hat{\omega})$ & $\text{Bias}(\hat{\eta})$ & $\text{MSE}(\hat{\omega})$ & $\text{MSE}(\hat{\eta})$  \\ \hline
 30 & 0.0287  & 0.1753 & 0.0100    & 0.2702 & 30&	0.0480  & 0.1463 & 0.0398    & 0.2625  \\
50 & 0.0203  & 0.1507 & 0.0059    & 0.2480 & 50&	 0.0337  & 0.1311 & 0.0240    & 0.2396  \\
100 & 0.0118  & 0.1398 & 0.0029    & 0.2183 & 100&	 0.0189  & 0.1165 & 0.0120    & 0.2069  \\
200 & 0.0073  & 0.1219 & 0.0015    & 0.1807 & 200&	 0.0122  & 0.0985 & 0.0063    & 0.1674 \\
300 & 0.0057  & 0.1133 & 0.0010    & 0.1633 & 300&	 0.0093  & 0.0955 & 0.0044    & 0.1455  \\
600 & 0.0023  & 0.1043 & 0.0005    & 0.1236 & 600&	 0.0033  & 0.0822 & 0.0022    & 0.1049  \\  \hline

\multicolumn{5}{|c|}{$\omega=0.5$, $\eta=1$} & \multicolumn{5}{|c|}{$\omega=1$, $\eta=1$} \\ \hline

30 & 0.0320  & 0.2313 & 0.0098    & 0.7653  &30	& 0.0579  & 0.2410 & 0.0398    & 0.8409  \\
50 & 0.0237  & 0.1903 & 0.0059    & 0.7206  &50	 & 0.0413  & 0.2026 & 0.0238    & 0.7759  \\
100 & 0.0142  & 0.1751 & 0.0029    & 0.6441  &100	 & 0.0244  & 0.1929 & 0.0115    & 0.6830  \\
200 & 0.0087  & 0.1548 & 0.0014    & 0.5593 &200 	 & 0.0153  & 0.1732 & 0.0061    & 0.5681  \\
300 & 0.0067  & 0.1498 & 0.0009    & 0.5105  &300	 & 0.0120  & 0.1627 & 0.0041    & 0.5127  \\
600 & 0.0031  & 0.1436 & 0.0004    & 0.4223  &400	 & 0.0049  & 0.1538 & 0.0020    & 0.3952  \\ \hline

\multicolumn{5}{|c|}{$\omega=0.5$, $\eta=1.5$} & \multicolumn{5}{|c|}{$\omega=1$, $\eta=1.5$} \\ \hline
30 & 0.0331  & 0.1091 & 0.0097    & 1.1293  &	30 & 0.0637  & 0.0566 & 0.0401    & 1.1088  \\
50 & 0.0248  & 0.0571 & 0.0059    & 1.0913  &	50 & 0.0474  & 0.0300 & 0.0243    & 1.0424  \\
100 & 0.0157  & 0.0280 & 0.0029    & 0.9955  &	100 & 0.0282  & 0.0250 & 0.0116    & 0.9098  \\
200 & 0.0097  & 0.0248 & 0.0014    & 0.8718  &	200 & 0.0177  & 0.0233 & 0.0058    & 0.7683  \\
300 & 0.0074  & 0.0187 & 0.0009    & 0.7911  &	300 & 0.0144  & 0.0148 & 0.0038    & 0.6933  \\
600 & 0.0038  & 0.0108 & 0.0004    & 0.6626  &	600 & 0.0067  & 0.0116 & 0.0018    & 0.5503  \\  \hline

\multicolumn{5}{|c|}{$\omega=0.5$, $\eta=2$} & \multicolumn{5}{|c|}{$\omega=1$, $\eta=2$} \\ \hline

30 & 0.0338  & 0.1501 & 0.0097    & 2.0222  &	30 & 0.0648  & 0.1041 & 0.0390    & 1.9827  \\
50 & 0.0253  & 0.0812 & 0.0059    & 1.9704  &	50 & 0.0485  & 0.0498 & 0.0239    & 1.8855  \\
100 & 0.0156  & 0.0434 & 0.0029    & 1.8135  &	100 & 0.0294  & 0.0480 & 0.0114    & 1.6686  \\
200 & 0.0099  & 0.0318 & 0.0014    & 1.6171  &	200 & 0.0184  & 0.0445 & 0.0056    & 1.4372  \\
300 & 0.0075  & 0.0173 & 0.0009    & 1.4908  &	300 & 0.0142  & 0.0387 & 0.0037    & 1.2858  \\
600 & 0.0039  & 0.0157 & 0.0004    & 1.2554  &	600 & 0.0071  & 0.0345 & 0.0017    & 1.0501  \\ \hline

\end{tabular}}
\end{table}

\section{Applications}
In this section, four real data sets are analyzed to demonstrate the applicability of the proposed Shiha distribution and compare it with some well-known distributions.

\noindent  \textbf{Data set 1:}
The first data set discussed by \cite{murt}, are the failure times of eight
components at three different temperatures 100, 120, 140. The data are:
14.712, 32.644, 61.979, 65.521, 105.50, 114.60, 120.40, 138.50, 8.610, 11.741, 54.535,
55.047, 58.928, 63.391, 105.18, 113.02, 2.998, 5.016, 15.628, 23.040, 27.851, 37.843,
38.050, 48.226

\noindent  \textbf{Data set 2:} The data set contains 34 measurements of vinyl chloride concentrations (mg/l) from clean up-gradient monitoring wells \cite{bhau}:
5.1, 1.2, 1.3, 0.6, 0.5, 2.4, 0.5, 1.1, 8, 0.8, 0.4, 0.6, 0.9, 0.4, 2, 0.5, 5.3, 3.2, 2.7, 2.9, 2.5, 2.3, 1,
0.2, 0.1, 0.1, 1.8, 0.9, 2, 4, 6.8, 1.2, 0.4, 0.2

\noindent  \textbf{Data set 3:}
The data set consists of 59 annual maximum precipitation observations recorded in Karachi, Pakistan, covering the period 1950–2009 \cite{phat}:
117.6, 157.7, 148.6, 11.4, 5.6, 63.6, 62.4, 11.8, 6.5, 54.9, 39.9, 16.8, 30.2, 38.4, 76.9, 73.4,
85, 256.3, 24.9, 148.6, 160.5, 131.3, 77, 155.2, 217.2, 105.5, 166.8, 157.9, 73.6, 291.4,
210.3, 315.7, 107.7, 33.3, 302.6, 159.1, 78.7, 33.2, 52.2, 92.7,150.4, 43.7, 68.3, 20.8,
179.4, 245.7, 19.5, 30, 270.4, 160, 96.3, 185.7, 429.3, 184.9, 262.5, 80.6, 138.2, 28, 39.3

 \noindent \textbf{Data set 4:}
The following observations represent the failure times (in minutes) obtained from a sample of 15 electronic components subjected to an accelerated life test \cite{law}:
1.4, 5.1, 6.3, 10.8, 12.1, 18.5, 19.7, 22.2, 23, 30.6, 37.3, 46.3, 53.9, 59.8, 66.2

%A brief description of the data and their main characteristics is provided below. The descriptive statistics, namely minimum (Min), %first quartile (Q1), median, mean, third quartile (Q3), maximum (Max), skewness (SK), and kurtosis (Kur), are summarized in %Table~\ref{tab:desc}
A brief description of the data and their main characteristics is provided in Table~\ref{tab:desc}
\begin{table}[!ht]
\centering
\caption{Summary statistics for the considered data sets }
\label{tab:desc}
\begin{tabular}{c c c c c c c c c c}
\hline
Data& Min & Q1 & Median & Q3 & Max & Mean & Variance & Sk & Kur \\ \hline
1& 2.998 & 21.187 & 51.38 & 75.436 & 138.5 & 55.123 & 1685.495 & 0.555 & 5.108 \\ \hline
2&0.1 & 0.5 & 1.15 & 2.475 & 8 & 1.879 & 3.813 & 1.604 & 8.005 \\ \hline
3 & 5.6 & 39.6 & 92.7 & 160.25 & 429.3 & 118.397 & 8688.987 & 1.024 & 6.766 \\ \hline
4& 1.4 & 11.45 & 22.2 & 41.8 & 66.2 & 27.547 & 431.117 & 0.566 & 5.06 \\ \hline
\end{tabular}
\end{table}

As an initial nonparametric check, we employ the total time on test (TTT) transform plot for each dataset to assess the suitability of the proposed distribution. We also examine the kernel density plot, box plot, and violin plot. From Figures~\ref{fig:ttt1}-\ref{fig:ttt3}, it is evident that all data sets are right-skewed. Moreover, data sets 1, 3, and 4 exhibit an increasing failure rate, while data set 2 shows a decreasing failure rate. These features suggest that the Shiha distribution is an appropriate model for fitting the data.
\begin{figure}[H]
  \centering
  \includegraphics[width=13cm, height=3.6cm]{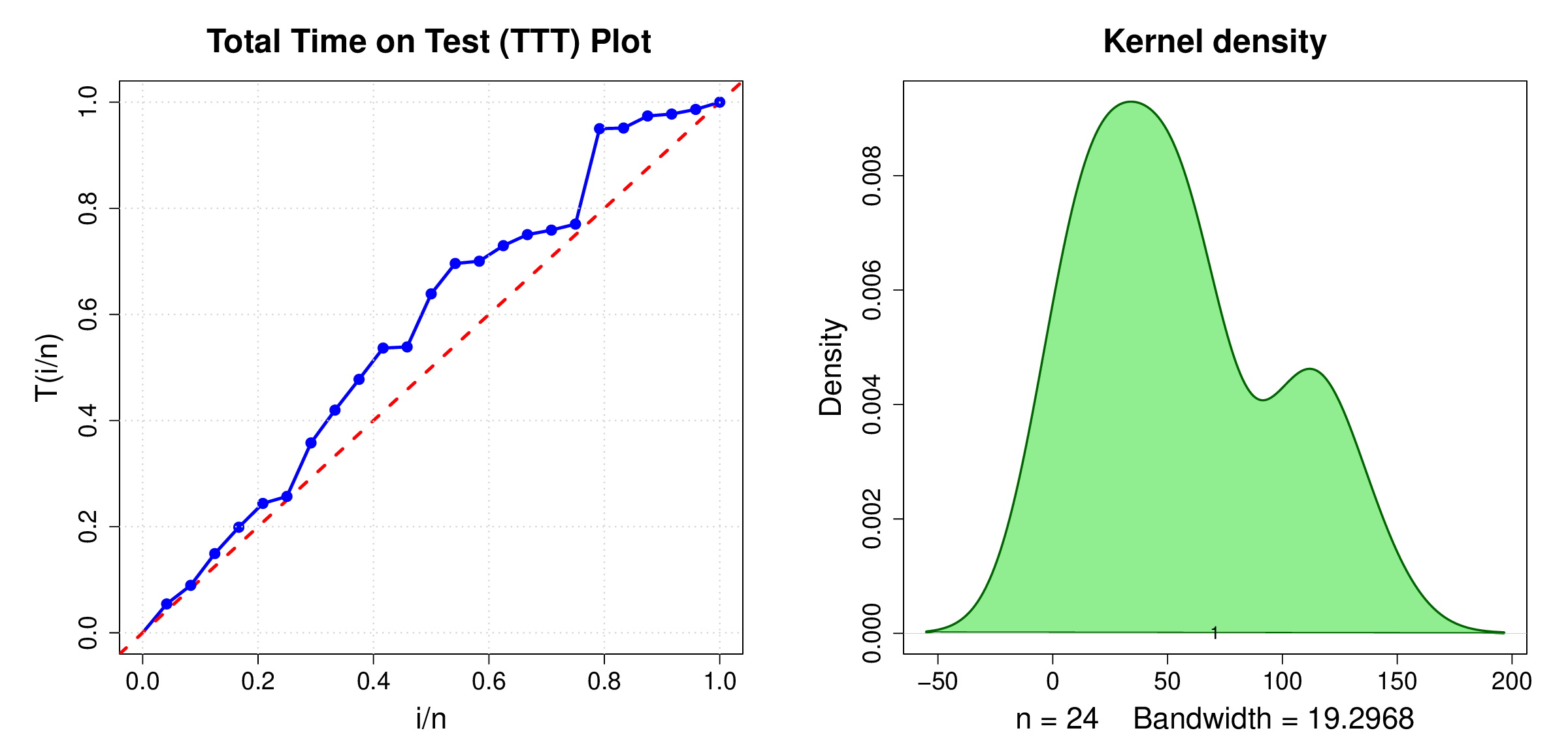}
  \includegraphics[width=12.5cm, height=3.4cm]{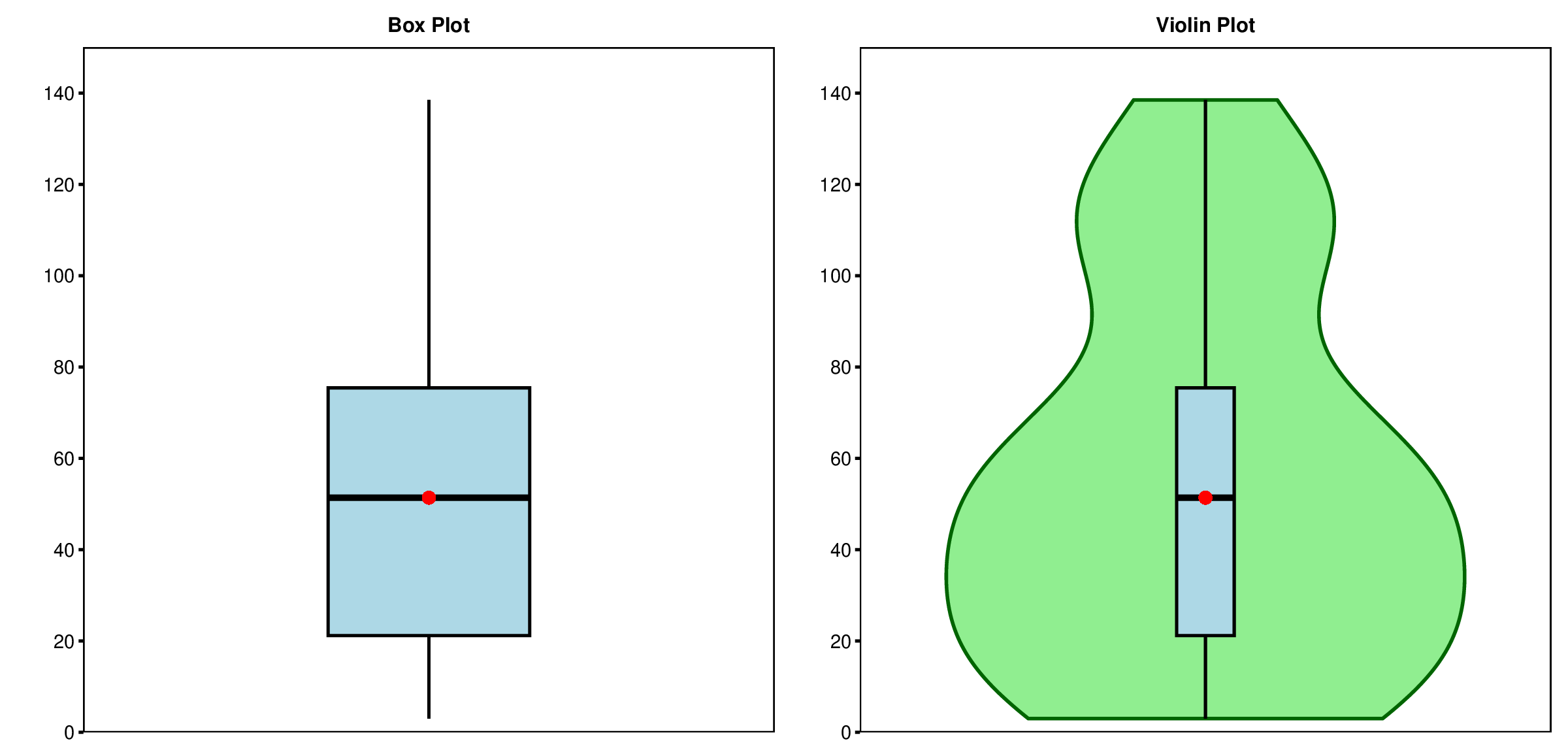}
  \caption{Graphical representation of data set 1.}
 \label{fig:ttt1}
\end{figure}
\begin{figure}[H]
  \centering
  \includegraphics[width=13cm, height=3.6cm]{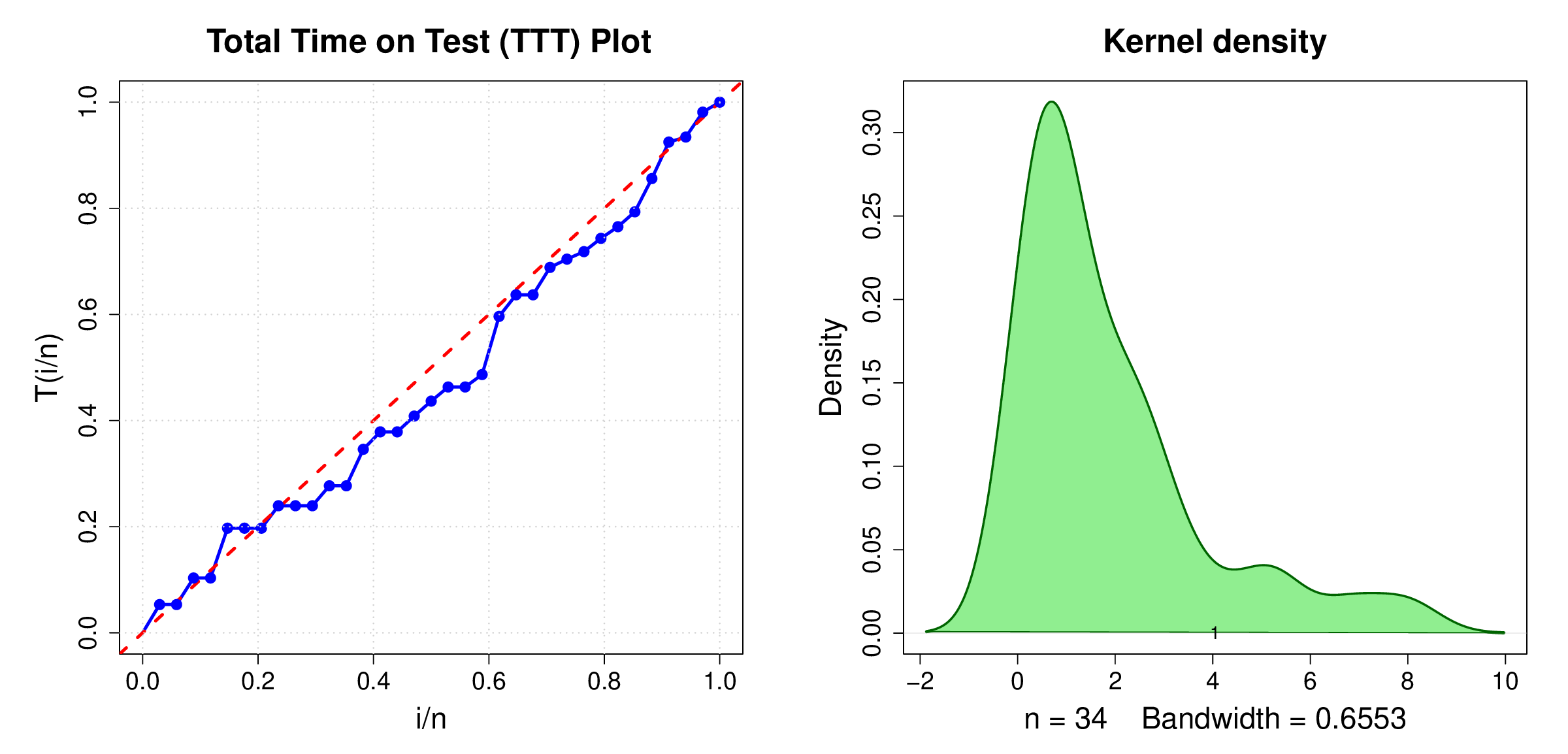}
  \includegraphics[width=12.5cm, height=3.4cm]{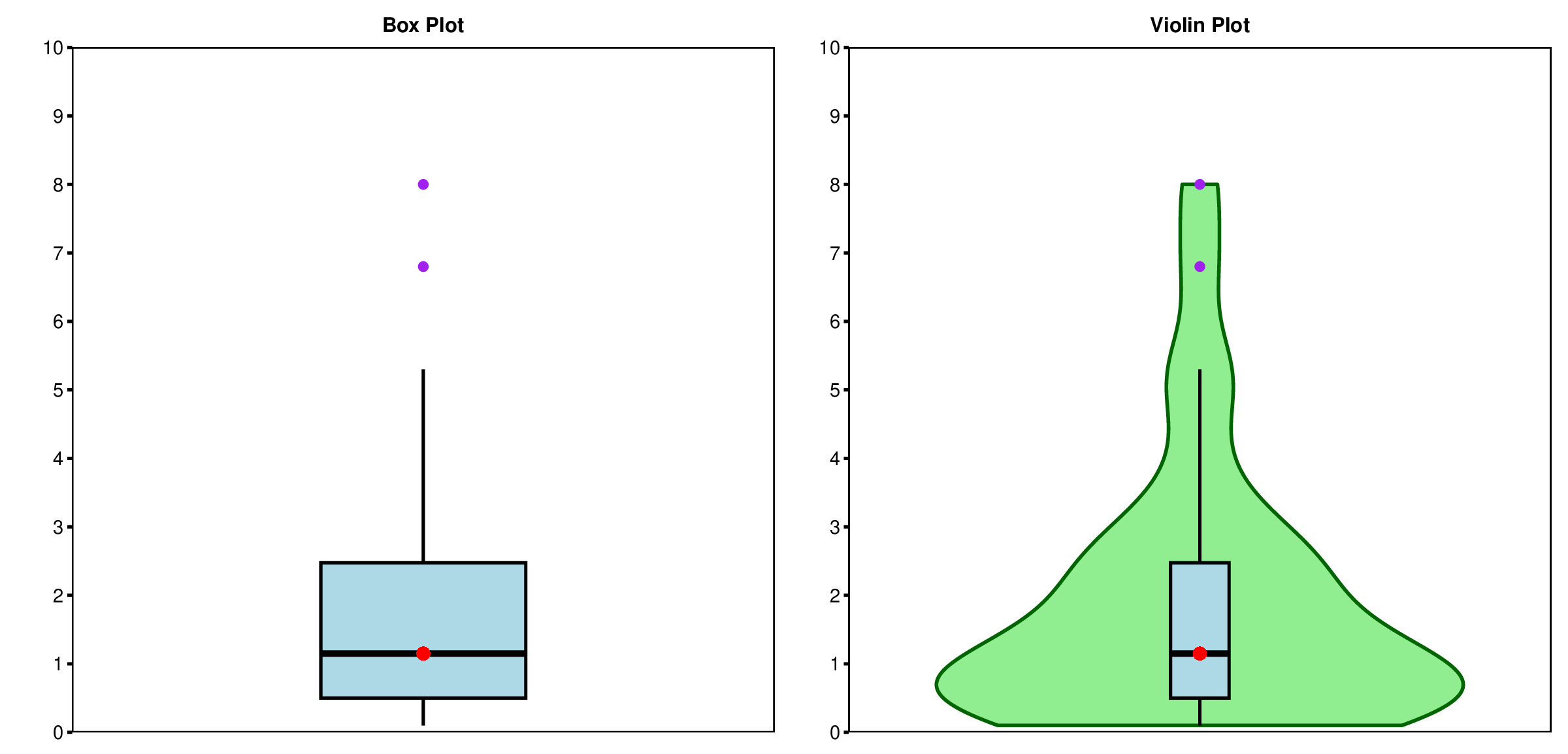}
  \caption{Graphical representation of data set 2.}
 \label{fig:ttt2}
\end{figure}
\begin{figure}[H]
  \centering
  \includegraphics[width=13cm, height=3.6cm]{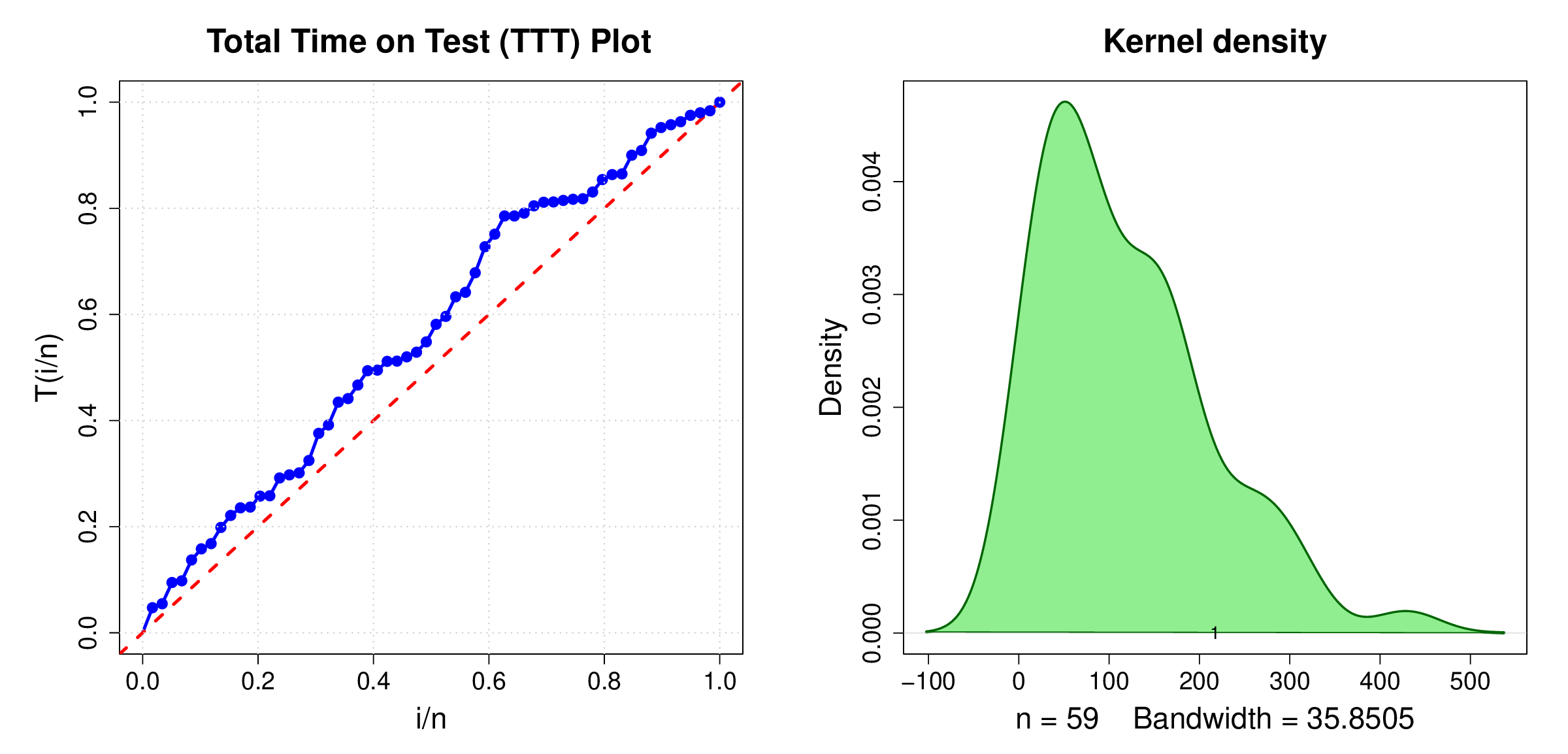}
  \includegraphics[width=12.5cm, height=3.4cm]{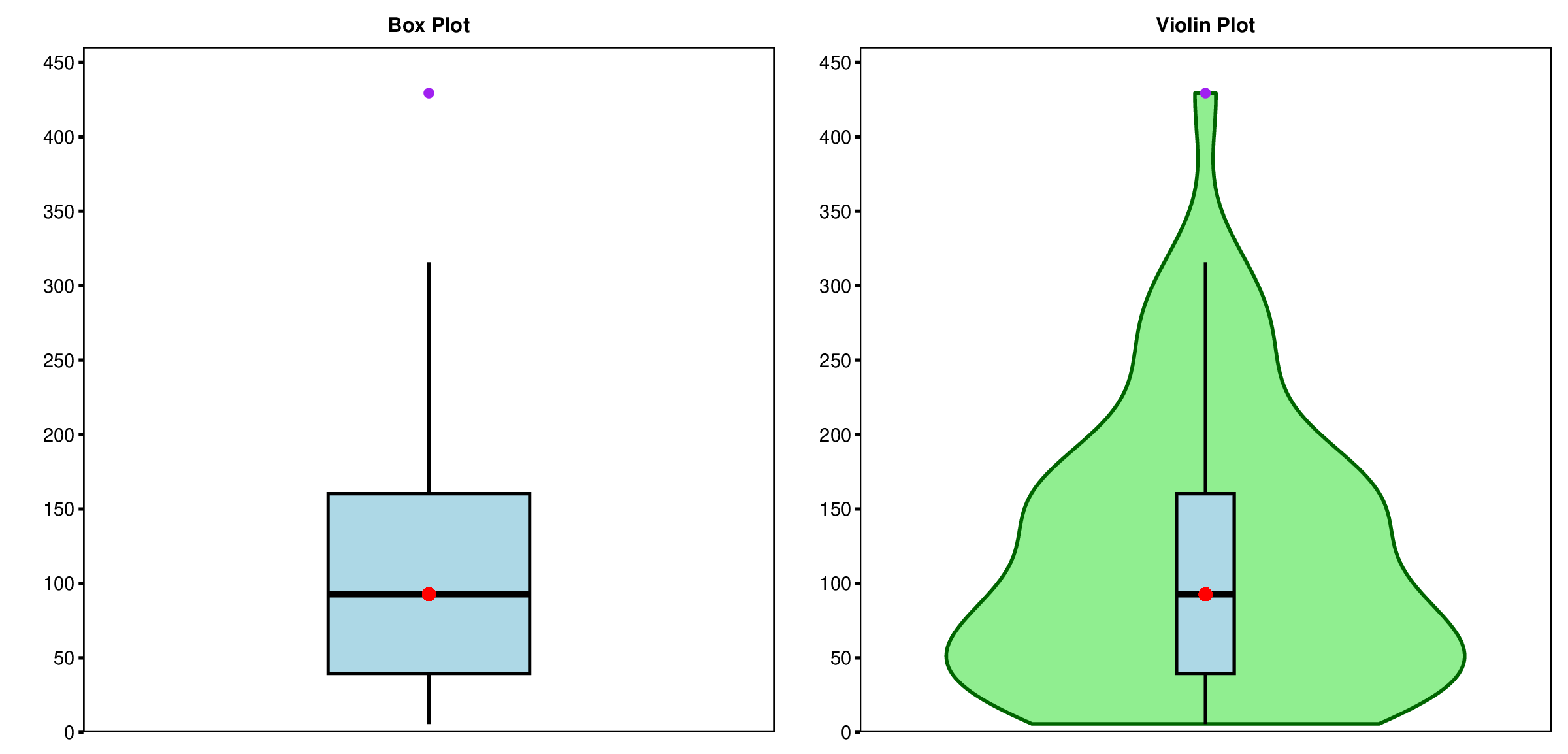}
  \caption{Graphical representation of data set 3.}
 \label{fig:ttt3}
\end{figure}
\begin{figure}[H]
  \centering
  \includegraphics[width=13cm, height=3.6cm]{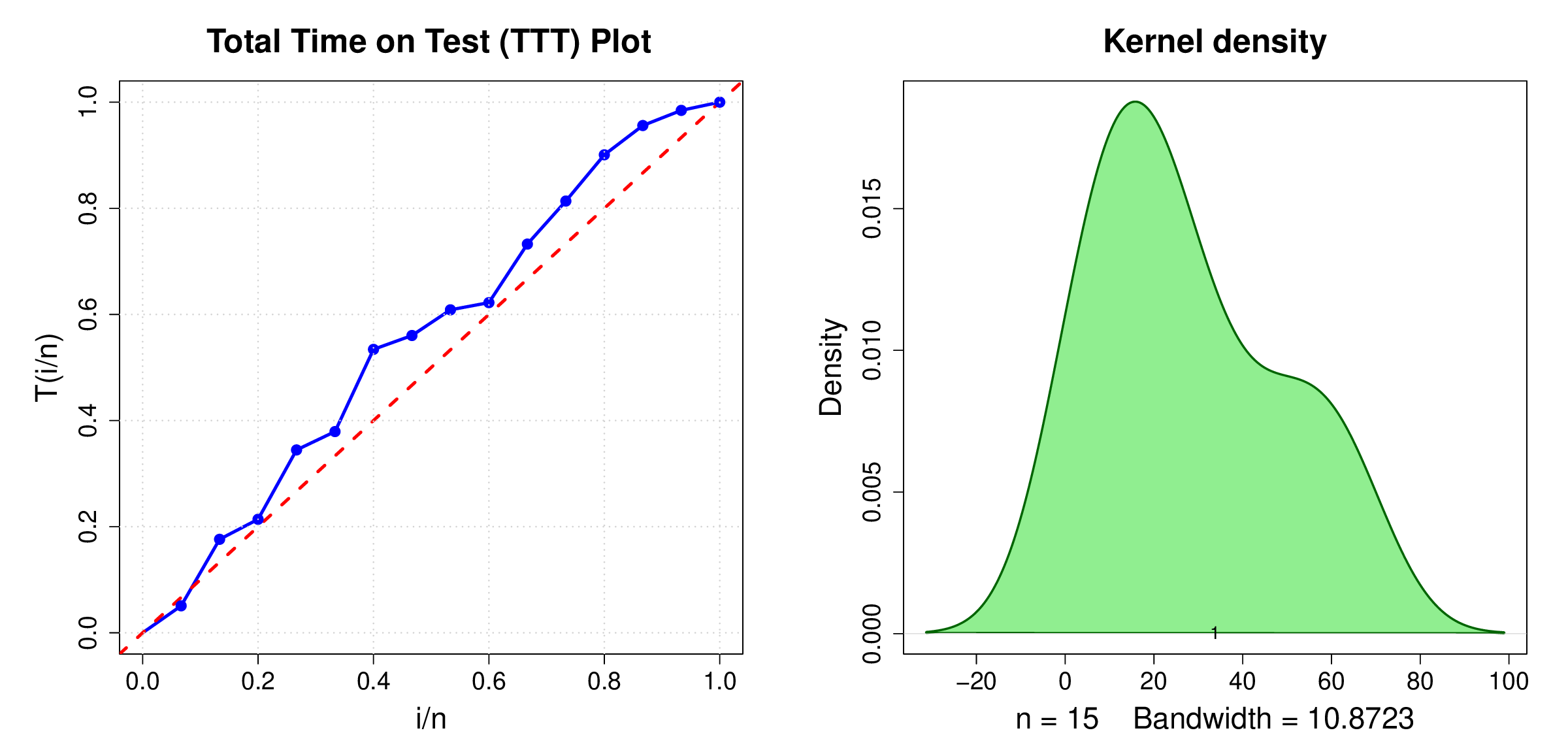}
  \includegraphics[width=12.5cm, height=3.4cm]{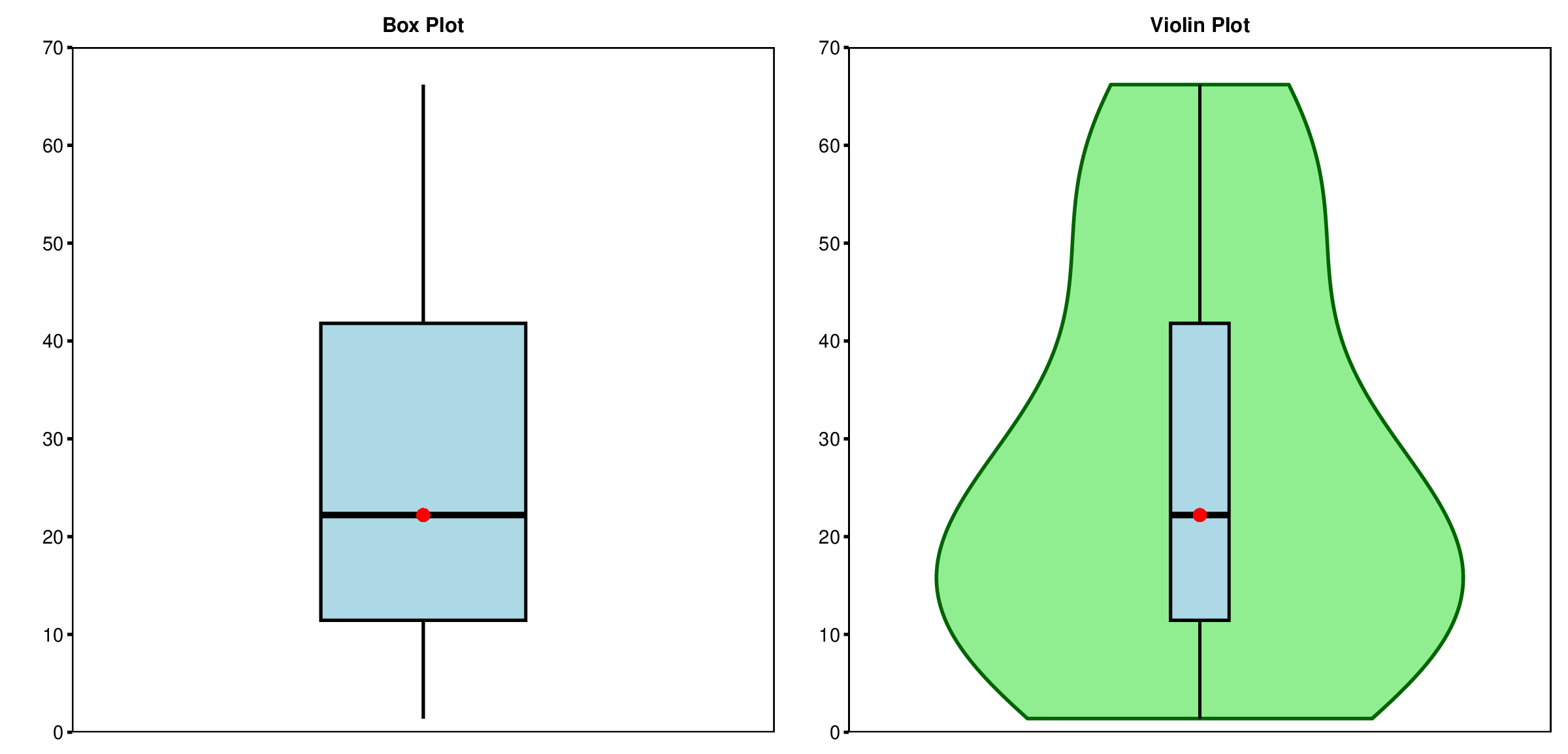}
  \caption{Graphical representation of data set 4.}
 \label{fig:ttt4}
\end{figure}
To examine whether the Shiha distribution provides an adequate fit to the observed data, we apply several goodness-of-fit tests. The Kolmogorov--Smirnov (K-S) statistic, which measures the maximum difference between the empirical and theoretical distribution functions. Since the K-S test is less sensitive in the tails, we also employ the Anderson--Darling (A-D) test, which gives more weight to deviations in the distribution tails. For this purpose, model selection is based on likelihood-based measures evaluated at the maximum likelihood estimates (MLEs). Specifically, the Akaike information criterion AIC $= 2k-2 \log L$  and Bayesian information criterion BIC $= k \log n-2 \log L $ are used, where $log L$ denotes the log-likelihood function at the MLE of the parameters, $k$ is the number of parameters, and $n$ is the sample size. The model achieving the lowest values of these statistics is regarded as the best-fitting candidate. To further assess the performance of the proposed Shiha distribution, it is compared with several well-known lifetime distributions:

\begin{enumerate}
  \item Alpha power transformed xgamma distribution (APTXGD) \cite{shukla}
\begin{equation}
f(y;\omega, \eta) =
\begin{cases}
\frac{\log \eta}{\eta - 1} \frac{\omega^2}{(1 + \omega)} \left( 1 + \frac{\omega}{2} y^2 \right) e^{-\omega y} \eta^{1 - u_0}, & \eta > 0, \, \eta \neq 1 \\
\frac{\omega^2}{(1 + \omega)} \left( 1 + \frac{\omega}{2} y^2 \right) e^{-\omega y}, & \eta = 1
\end{cases}
\end{equation}
where
\[
 u_0 = \frac{\left( 1 + \omega + \omega y + \frac{\omega^2 y^2}{2} \right)}{(1 + \omega)} e^{-\omega y},\quad y > 0, \, \omega > 0.
\]
  \item Power Lindley distribution (PLD) \cite{ghit}
\begin{equation}
f(y;\omega,\eta) = \frac{\eta \, \omega^{2}}{1+\omega}
\left( 1 + y^{\eta} \right) y^{\eta-1} e^{-\omega y^{\eta}},
\quad y > 0, \; \omega > 0, \; \eta > 0,
\end{equation}
  \item The three parameter generalized Lindley distribution (TPGLD) \cite{ekho}
\begin{equation}
f(y;\omega,\eta, \alpha) = \frac{\alpha \, \omega^{2} \, (\eta + y^{\alpha}) \, y^{\alpha - 1} e^{-\omega y^{\alpha}}}{1 + \omega \eta},
\quad y > 0, \omega, \eta > 0, \alpha>0
\end{equation}
\item Chris-Jerry distribution (CJD) \cite{chri}
\begin{equation}
f(y;\omega)=\frac{{\omega}^2}{\omega +2}(1+\omega \, y^2) e^{-\omega \, y}\quad , y>0 ,\,  \omega >0.
\end{equation}
\item Akash distribution (AKD) \cite{shanker1}
\begin{equation}
f(y;\omega)=\frac{{\omega}^3}{{\omega}^2 +2}(1+ y^2) e^{-\omega \, y}\quad , y>0 ,\,  \omega >0.
\end{equation}
\end{enumerate}
Tables~\ref{T:app1}–\ref{T:app4} present the MLEs of the parameters, together with the corresponding goodness-of-fit statistics for each fitted distribution across the four data sets. Additionally, the graphical goodness-of-fit assessments for the four data sets are shown in Figures~\ref{fig:t1}–\ref{fig:t4}. Each figure presents the histogram with the fitted densities, the theoretical and empirical cdfs, the QQ and PP plots for the six competing models.

\begin{table}[!ht]
\centering
\caption{Parameter estimates, and goodness-of-fit measures for the data set 1.}\label{T:app1}
\adjustbox{max width=\textwidth, center}{
\begin{tabular}{|l|  l|  l| c| c| c| c| c |}\hline
Model  & estimates   &AIC   &BIC     & A-D (p-value)  & K-S (p-value)     \\ \hline

\textbf{Shiha }& $\hat{\omega}=0.0152$, $\hat{\eta}=1.4689$  & \textbf{242.5978} & \textbf{244.9539} & \textbf{0.2988 (0.9385)} & \textbf{0.1210 (0.8326)} \\
APTXGD & $\hat{\omega}=0.0447$, $\hat{\eta}=0.3623$ & 246.9883 & 249.3444 & 1.2148 (0.2616) & 0.1665 (0.4693) \\
PLD & $\hat{\omega}=0.0602$, $\hat{\eta}=0.8744$  & 243.0938 & 245.4499 & 0.3691 (0.8776) & 0.1231 (0.8177) \\
TPGLD & $\hat{\omega}=0.0118$, $\hat{\eta}=1.1825$, $\hat{\alpha}=109.8619$ & 244.0922 & 247.6263 & 0.3134 (0.9271) & 0.1259 (0.7963) \\
CJD & $\hat{\omega}=0.0527$  & 247.4051 & 248.5831 & 1.8339 (0.1140) & 0.1862 (0.3340) \\
AKD &$\hat{\omega}=0.0544$  & 250.5197 & 251.6977 & 2.3536 (0.0597) & 0.1944 (0.2859) \\ \hline
\end{tabular}}
\end{table}

\begin{figure}[H]
  \centering
  \includegraphics[width=14cm, height=4.5cm]{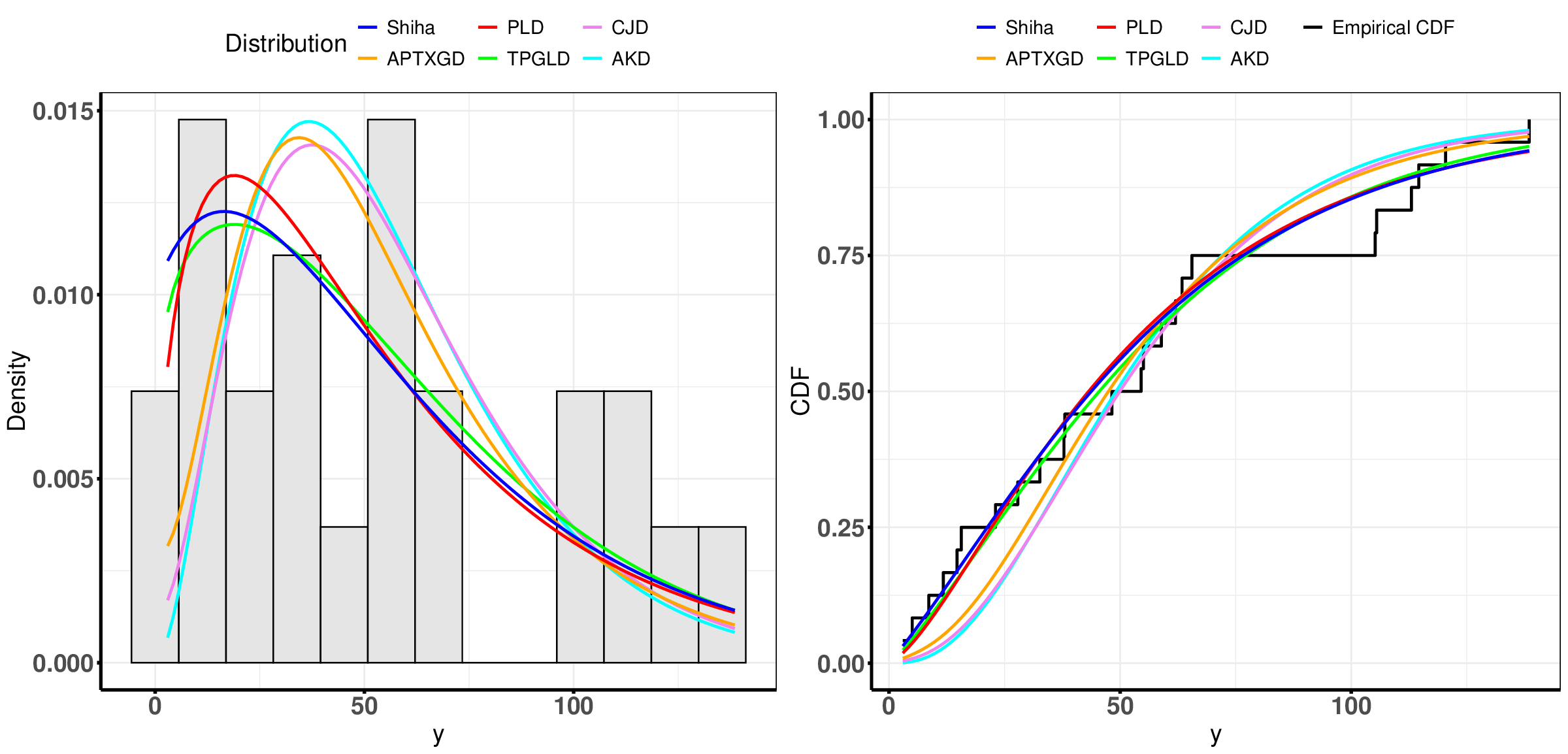}
  \includegraphics[width=14cm, height=4.5cm]{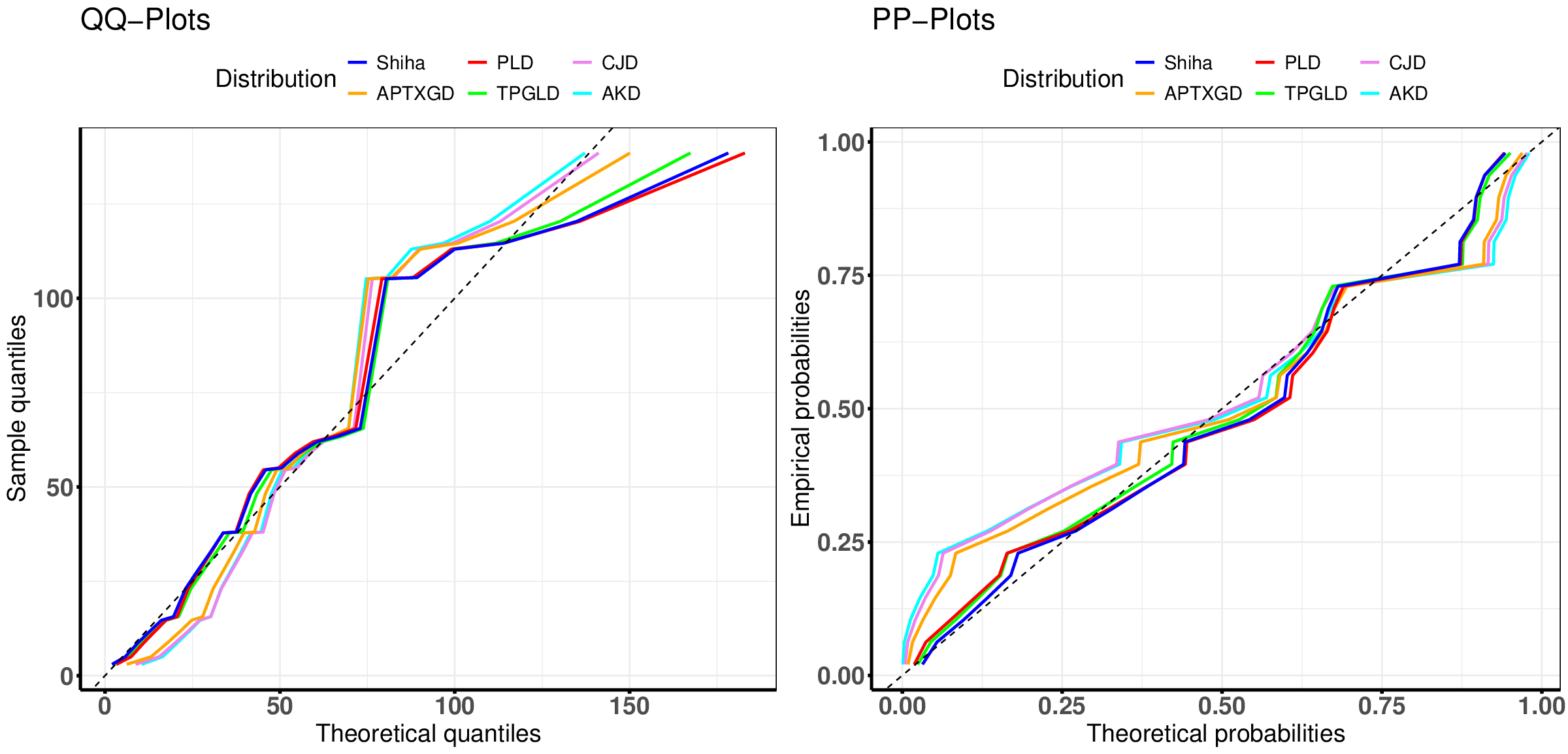}
  \caption{The fitted pdf, cdf, QQ, and PP for data set 1}
 \label{fig:t1}
\end{figure}

\begin{table}[!ht]
\centering
\caption{Parameter estimates, and goodness-of-fit measures for the data set 2.}\label{T:app2}
\adjustbox{max width=\textwidth, center}{
\begin{tabular}{|l|  l|  l| c| c| c| c| c |}\hline
Model  & estimates   &AIC   &BIC     & A-D (p-value)  & K-S (p-value)     \\ \hline
\textbf{Shiha}& $\hat{\omega}=0.532$, $\hat{\eta}=1e-04$  & \textbf{114.9055} & \textbf{117.9582} &\textbf{ 0.2720 (0.9573)} & \textbf{0.0890 (0.9506)} \\
APTXGD & $\hat{\omega}=0.6911$, $\hat{\eta}=0.1182$ & 115.2675 & 118.3202 & 0.2979 (0.9394) & 0.1012 (0.8774) \\
PL & $\hat{\omega}=0.9139$, $\hat{\eta}=0.8832$  & 115.5198 & 118.5726 & 0.3119 (0.9285) & 0.0944 (0.9226) \\
TPGLD & $\hat{\omega}=0.5272$, $\hat{\eta}=1.0102$, $\hat{\alpha}=1000$ & 116.8992 & 121.4783 & 0.2826 (0.9504) & 0.0918 (0.9366) \\
CJD & $\hat{\omega}=1.1644$  & 117.8536 & 119.3800 & 1.2426 (0.2517) & 0.1783 (0.2301) \\
AKD & $\hat{\omega}=1.1657$  & 117.1493 & 118.6756 & 0.9436 (0.3876) & 0.1564 (0.3762) \\ \hline

\end{tabular}}
\end{table}

\begin{figure}[H]
  \centering
  \includegraphics[width=14cm, height=4.5cm]{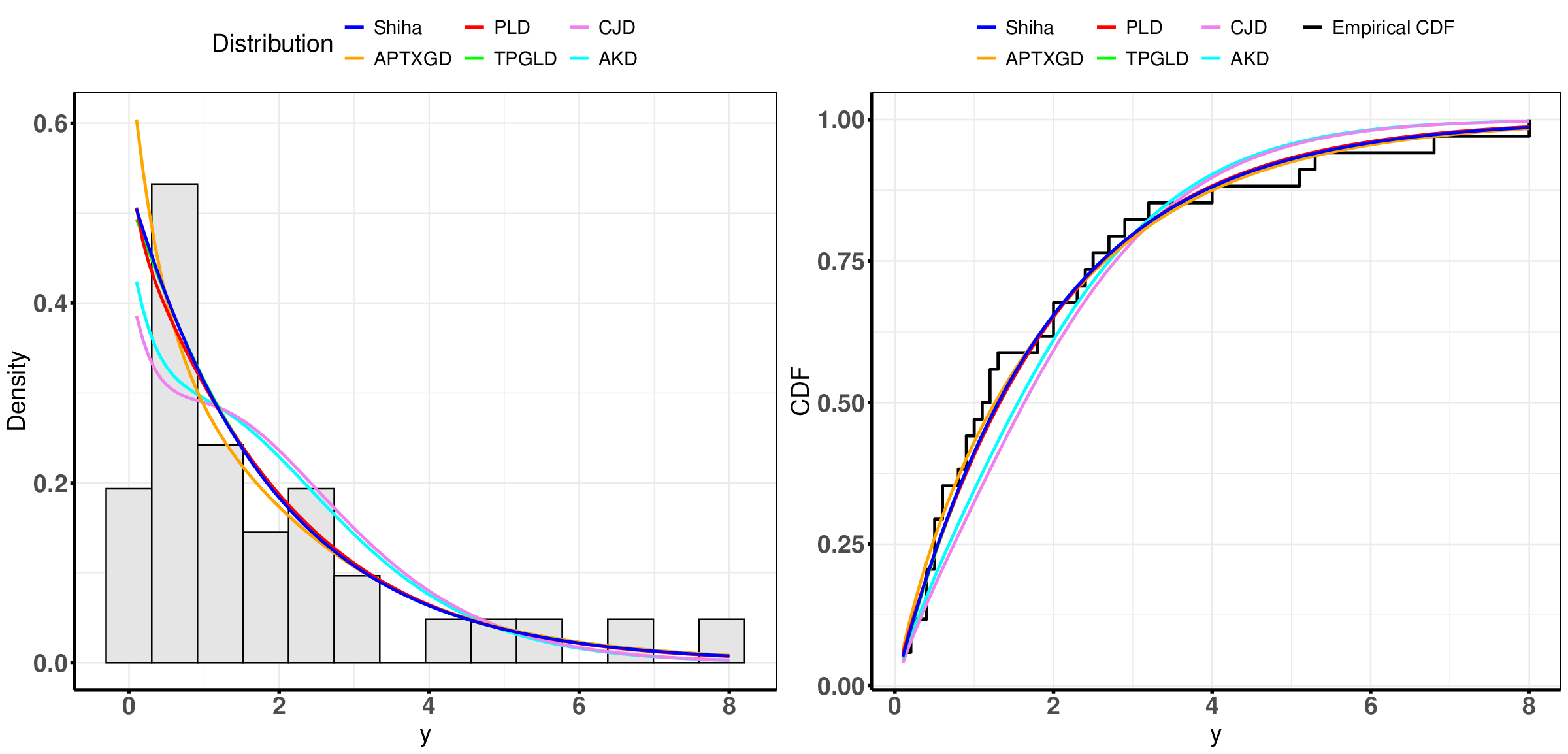}
   \includegraphics[width=14cm, height=4.5cm]{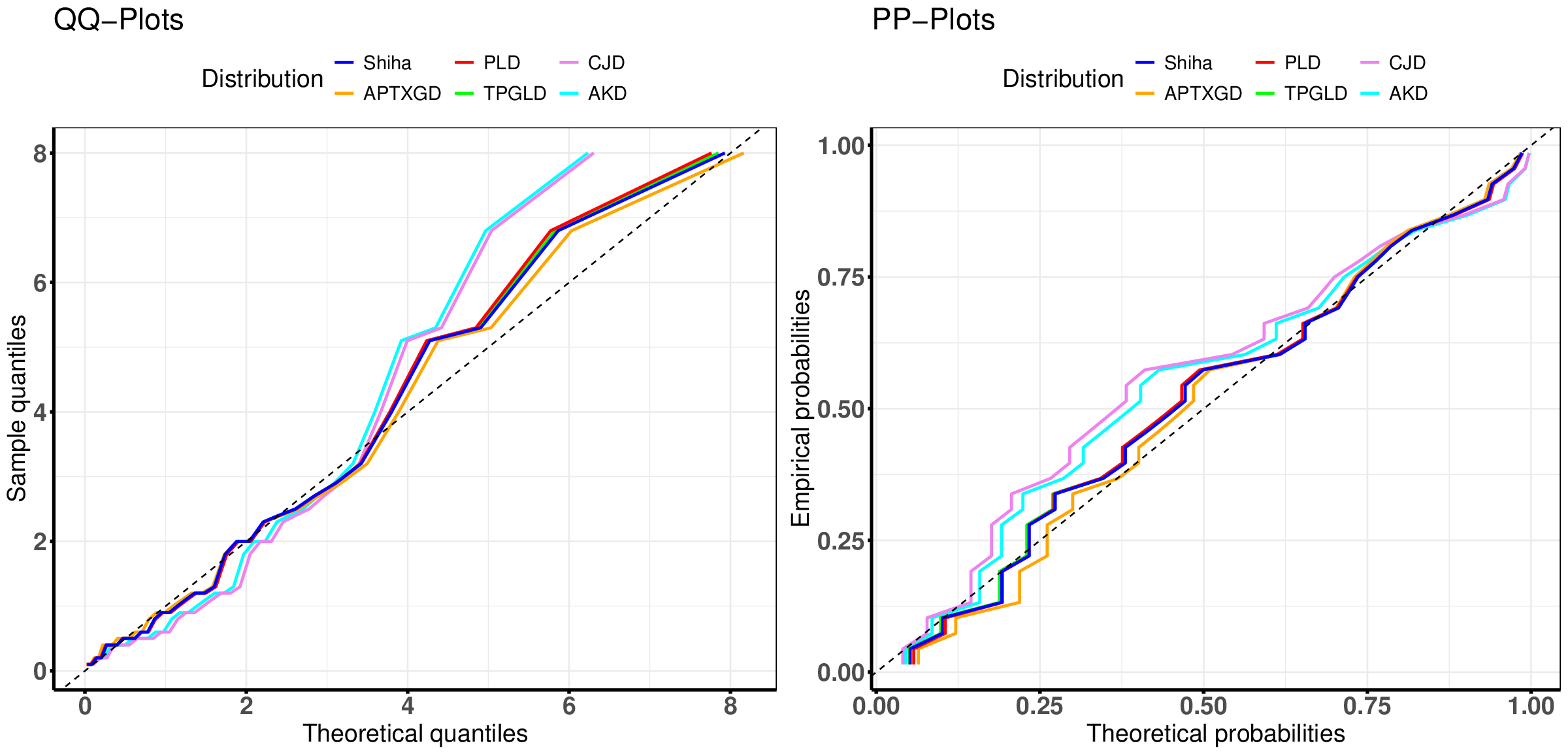}
  \caption{The fitted pdf, cdf, QQ, and PP for data set 2}
 \label{fig:t2}
\end{figure}

\begin{table}[!ht]
\centering
\caption{Parameter estimates, and goodness-of-fit measures for the data set 3.}\label{T:app3}
\adjustbox{max width=\textwidth, center}{
\begin{tabular}{|l|  l|  l| c| c| c| c| c |}\hline
Model  & estimates   &AIC   &BIC     & A-D (p-value)  & K-S (p-value)     \\ \hline
\textbf{Shiha} & $\hat{\omega}=0.0071$, $\hat{\eta}=0.0586$  & \textbf{681.6567} & \textbf{685.8118} & \textbf{0.2854 (0.9485)} & \textbf{0.0973 (0.6317)} \\
APTXGD & $\hat{\omega}=0.0207$, $\hat{\eta}=0.2844$  & 696.3373 & 700.4924 & 3.3364 (0.0187) & 0.1669 (0.0748) \\
PLD & $\hat{\omega}=0.0354$, $\hat{\eta}=0.8498$  & 681.7449 & 685.9000 & 0.3749 (0.8728) & 0.1057 (0.5255) \\
TPGLD & $\hat{\omega}=0.0389$, $\hat{\eta}=0.8333 $, $\hat{\alpha}=0.231$ & 683.6996 & 689.9322 & 0.3870 (0.8610) & 0.1069 (0.5106) \\
CJD & $\hat{\omega}=0.0249$ & 699.7348 & 701.8123 & 4.6818 (0.0041) & 0.1852 (0.0350) \\
AKD & $\hat{\omega}=0.0253$  & 703.7940 & 705.8715 & 5.2981 (0.0021) & 0.1886 (0.0301) \\  \hline
\end{tabular}}
\end{table}

\begin{figure}[H]
  \centering
  \includegraphics[width=14cm, height=4.5cm]{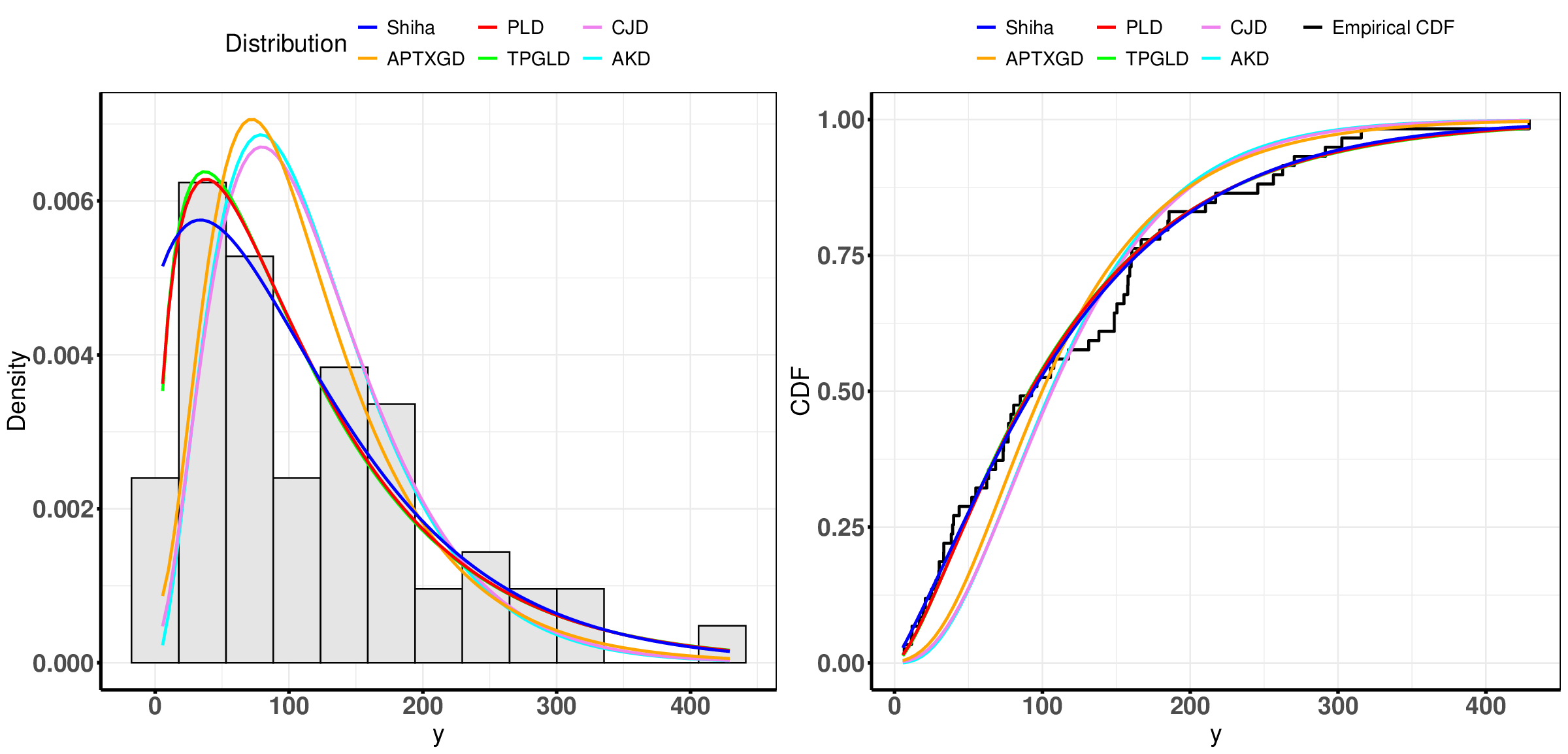}
   \includegraphics[width=14cm, height=4.5cm]{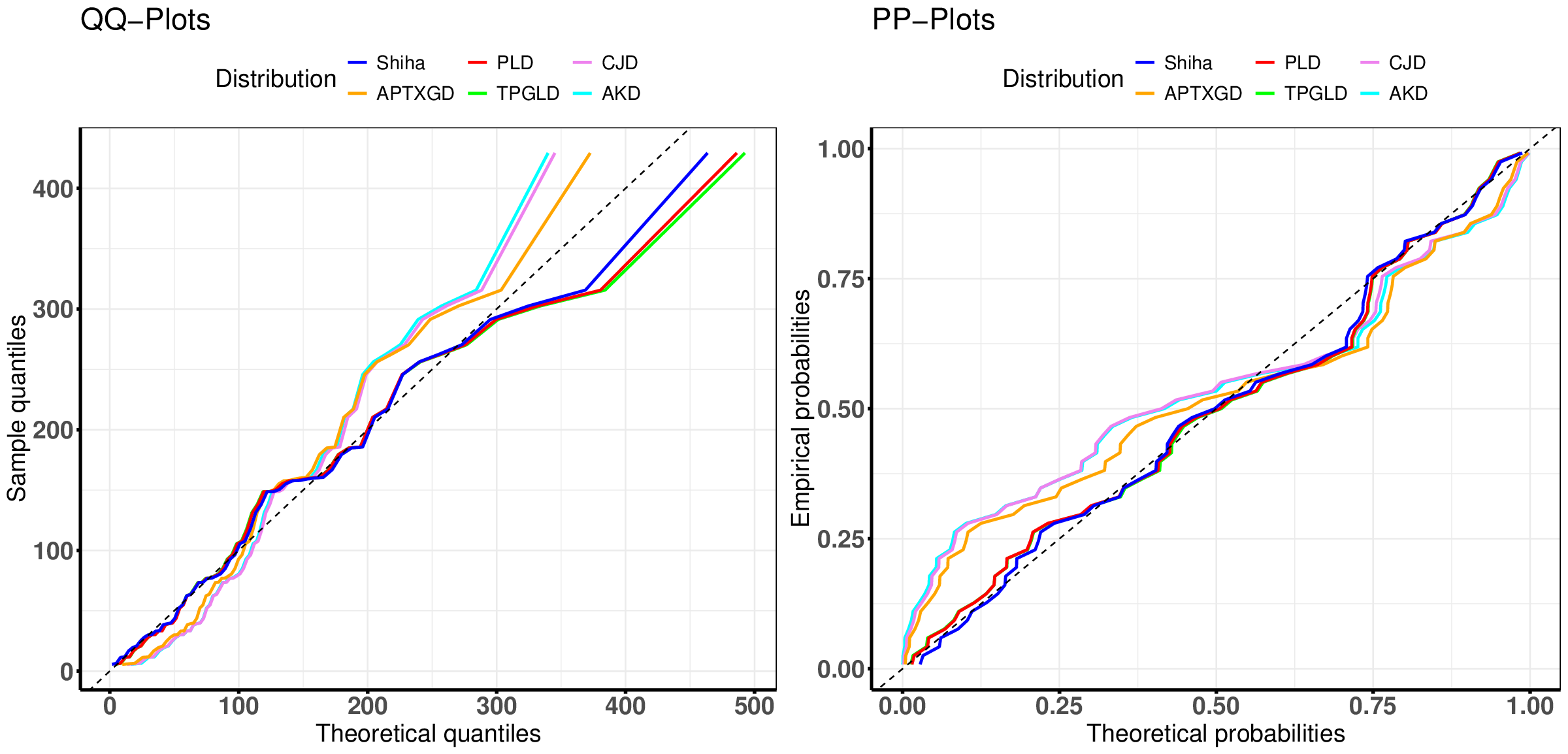}
  \caption{The fitted pdf, cdf, QQ, and PP for data set 3}
 \label{fig:t3}
\end{figure}

\begin{table}[!ht]
\centering
\caption{Parameter estimates, and goodness-of-fit measures for the data set 4.}\label{T:app4}
\adjustbox{max width=\textwidth, center}{
\begin{tabular}{|l|  l|  l|  c| c| c| c |}\hline
Model  & estimates    &AIC   &BIC     & A-D (p-value)  & K-S (p-value)     \\ \hline
\textbf{Shiha }& $\hat{\omega}=0.0304$, $\hat{\eta}=5.6774$  & \textbf{132.2349} & \textbf{133.6510} & \textbf{0.1661 (0.9973)} & 0.0979 (0.9958) \\
APTXGD & $\hat{\omega}=0.0861$, $\hat{\eta}=0.3527$  & 133.4417 & 134.8578 & 0.4649 (0.7801) & 0.1342 (0.9169) \\
PLD & $\hat{\omega}=0.0977$, $\hat{\eta}=0.9043$ & 132.4768 & 133.8929 & 0.1950 (0.9920) & 0.0986 (0.9954) \\
TPGLD & $\hat{\omega}=0.02231$, $\hat{\eta}= 1.2231$, $\hat{\alpha}=71.4391$  & 133.9388 & 136.0629 & 0.1819 (0.9948) & \textbf{0.0971 (0.9962)} \\
CJD & $\hat{\omega}=0.1034$  & 132.9871 & 133.6951 & 0.8732 (0.4291) & 0.1731 (0.6978) \\
AKD & $\hat{\omega}=0.1085$  & 135.6842 & 136.3922 & 1.2757 (0.2400) & 0.1841 (0.6248) \\
 \hline
\end{tabular}}
\end{table}
\begin{figure}[H]
  \centering
  \includegraphics[width=14cm, height=4cm]{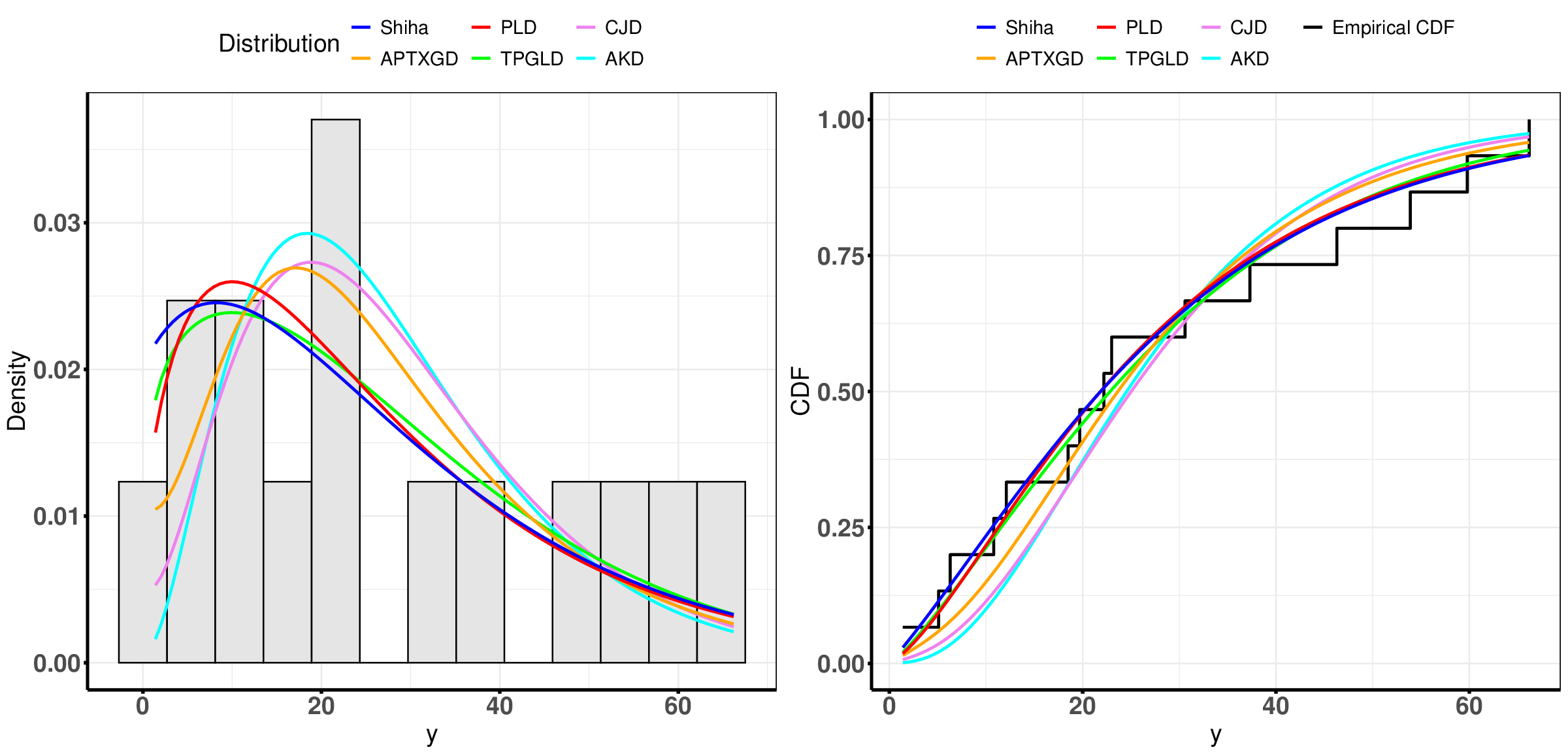}
  \includegraphics[width=14cm, height=4.5cm]{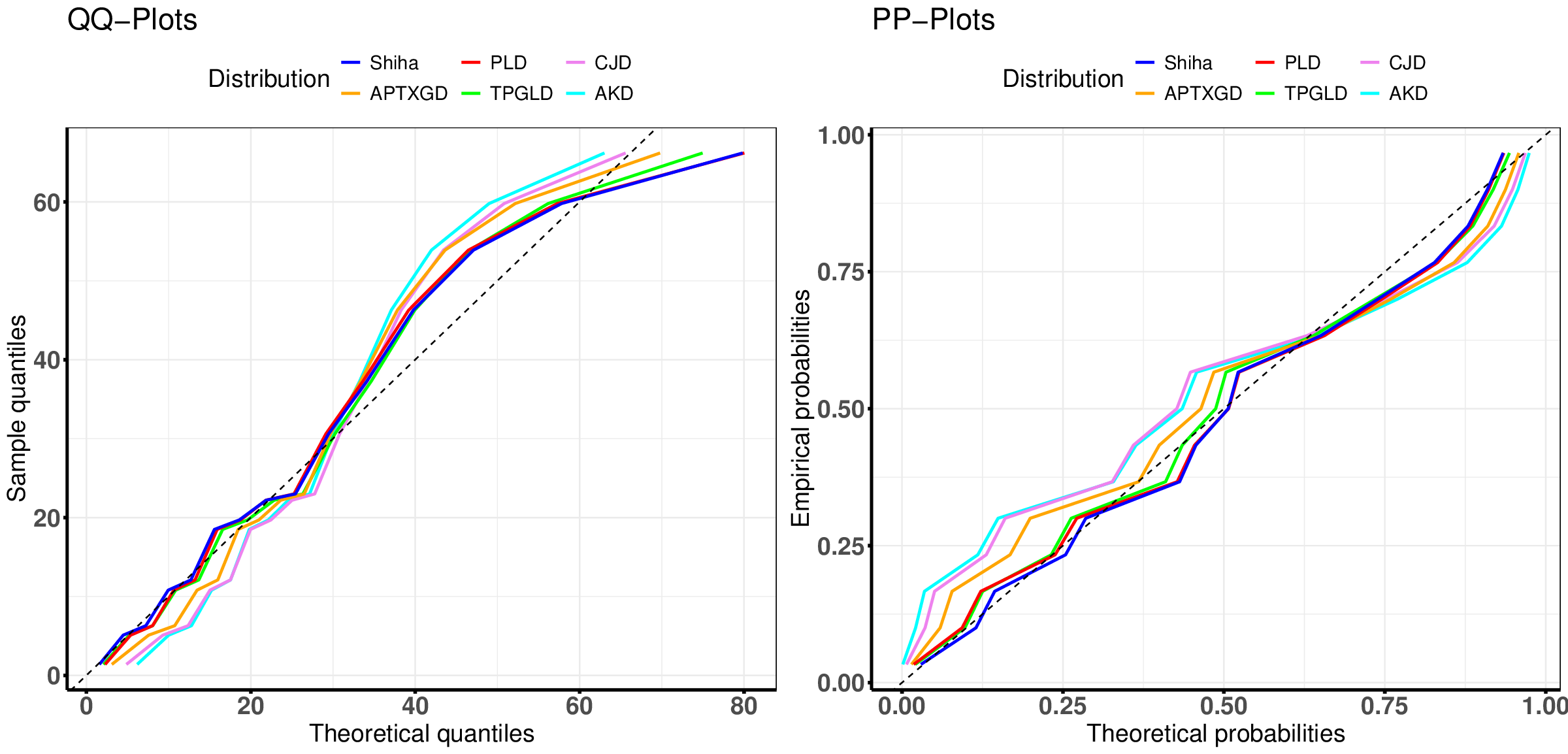}
  \caption{Graphical representation of dataset 4.}
 \label{fig:t4}
\end{figure}
Based on the results reported in Tables~\ref{T:app1}–\ref{T:app4}, the Shiha distribution exhibits lower values of AIC, BIC, A–D, and K–S, as well as higher A–D and K–S p-values. Together with the visual evidence shown in Figures~\ref{fig:t1}–\ref{fig:t4}, which demonstrate that the Shiha distribution closely matches the empirical data, we conclude that the Shiha model provides the best fit compared with the APTXGD, PLD, TPGLD, CJD, and AKD distributions across all data sets considered.
\section{Conclusions}
This paper introduces the Shiha distribution, a new and flexible lifetime model suitable for right-skewed data commonly observed in reliability engineering and environmental applications. The theoretical investigation covered several essential statistical properties, including moments, the hazard rate function, the quantile function, and entropy, revealing the distribution’s ability to capture diverse data shapes and tail behaviors. The stress--strength reliability analysis further demonstrated the applicability of the proposed model in reliability studies.
Parameter estimation was performed using the maximum likelihood method, and a Monte Carlo simulation study was conducted to examine the estimators’ performance. The results demonstrated the expected consistency of the MLEs, with both bias and mean squared error decreasing as the sample size increased.
The practical relevance  of the Shiha distribution was further verified through its application to four real datasets. Comparisons with several well-known lifetime models showed that the proposed distribution provides superior goodness-of-fit across all cases, confirming its effectiveness and flexibility.

% ------------------------------------------------------------------------

\section*{Conflict of interest}
The authors declare that there is no conflict of interests regarding the publication of this paper.

% ------------------------------------------------------------------------
\end{document}